\documentclass{article}

\usepackage{microtype}
\usepackage{graphicx}
\usepackage{subfigure}
\usepackage{booktabs} %

\usepackage{url}
\usepackage{graphicx}
\usepackage{braket}
\usepackage{enumerate}
\usepackage{color,graphics}

\usepackage{float}
\usepackage{amsmath,amsthm,amssymb,algorithmic,epsfig,graphicx, tikz}
\usepackage{mathtools}
\usepackage{multirow}
\usepackage{varioref}
\usepackage{bm}

\usepackage{hyperref}

\newcommand{\norm}[1]{\left\lVert#1\right\rVert}
\newcommand{\R}{\mathbb{R}}
\newcommand{\errgauss}{\epsilon_1}

\newcommand{\errlikelihood}{\epsilon_\tau}
\newcommand{\errmult}{\epsilon_3}
\newcommand{\errnorms}{\epsilon_{norm}}
\newcommand{\errtom}{\epsilon_{tom}}

\newtheorem{theorem}{Theorem}[section]
\newtheorem{lemma}[theorem]{Lemma}
\newtheorem{claim}[theorem]{Claim}

\newtheorem{corollary}[theorem]{Corollary}

\newtheorem*{theorem*}{Theorem}
\newtheorem*{result*}{Result}

\newtheorem{assumption*}{Assumption}
\newtheorem{definition}{Definition}
\DeclareMathOperator*{\argmax}{arg\,max}

\usepackage[accepted]{icml2020}

\icmltitlerunning{Quantum Expectation-Maximization for Gaussian mixture models}

\begin{document}

\twocolumn[
\icmltitle{Quantum Expectation-Maximization for Gaussian mixture models}

\icmlsetsymbol{equal}{*}

\begin{icmlauthorlist}
\icmlauthor{Iordanis Kerenidis}{irif,qcware}
\icmlauthor{Alessandro Luongo}{atos,irif}
\icmlauthor{Anupam Prakash}{irif,qcware}
\end{icmlauthorlist}

\icmlaffiliation{atos}{Atos Quantum Lab - Les Clayes sous bois}
\icmlaffiliation{irif}{IRIF - Paris}
\icmlaffiliation{qcware}{QCWare Corp.- Palo Alto}

\icmlcorrespondingauthor{Alessandro Luongo}{aluongo@irif.fr}

\icmlkeywords{Machine Learning, ICML}

\vskip 0.3in
]

\printAffiliationsAndNotice{\icmlEqualContribution} %

\begin{abstract}
    We define a quantum version of Expectation-Maximization (QEM), a fundamental tool in unsupervised machine learning, often used to solve Maximum Likelihood (ML) and Maximum A Posteriori (MAP) estimation problems. We use QEM to fit a Gaussian Mixture Model, and show how to generalize it to fit mixture models with base distributions in the exponential family. Given quantum access to a dataset, our algorithm has convergence and precision guarantees similar to the classical algorithm, while the runtime is polylogarithmic in the number of elements in the training set and polynomial in other parameters, such as the dimension of the feature space and the number of components in the mixture. We discuss the performance of the algorithm on datasets that are expected to be classified successfully by classical EM and provide guarantees for its runtime.
\end{abstract}

\section{Introduction}

Over the last few years, the effort to find real-world applications of quantum computers has greatly intensified. Along with chemistry, material sciences, finance, one of the fields where quantum computers are expected to be most beneficial is machine learning. Several different algorithms have been proposed for quantum machine learning \citep{Wiebe2017adversaries,HarrowHassidim2009HHL, subacsi2019quantum, farhi2018classification, montanaro2016quantum, biamonte2017quantum, chakrabarti2019quantum}, both for the supervised and unsupervised setting, and despite the lack of large-scale quantum computers and quantum memory devices, some quantum algorithms have been demonstrated in proof-of-principle experiments \citep{li2015experimental,otterbach2017unsupervised, jiang2019experimental}. 

Here, we look at Expectation-Maximization (EM), a fundamental algorithm in unsupervised learning, that can be used to fit different mixture models and give maximum likelihood estimates for \emph{latent variable models}. Such generative models are one of the most promising approaches to unsupervised problems. The goal of a generative model is to learn a probability distribution that is most likely to have generated the data collected in a dataset set $V \in \mathbb{R}^{n \times d}$ of $n$ vectors of $d$ features. Fitting a model consists in learning the parameters of a probability distribution $p$ in a certain parameterized family that best describes our vectors $v_i$. This formulation allows one to reduce a statistical problem into an optimization problem. For a given machine learning model $\gamma$, under the assumption that each point is independent and identically distributed the log-likelihood of a dataset $V$ is defined as $\ell(\gamma;V) := \sum_{i =1}^n \log p(v_i | \gamma) $, where $p(v_i|\gamma)$ is the probability that a point $v_i$ comes from model $\gamma$. For ML estimates we want to find the model $\gamma^*_{ML} := \argmax_{\gamma} \ell(\gamma;V)$. Due to the indented landscape of the function, optimizing $\ell$ using gradient based techniques often do not perform well. MAP estimates can be seen as the Bayesian version of maximum likelihood estimation problems, and are defined as $\gamma^*_{MAP} := \argmax_{\gamma} \sum_{i =1}^n \log p(v_i | \gamma) + \log p(\gamma)$. MAP estimates are often preferred over ML estimates, due to a reduced propensity to overfit.	

The EM algorithm has been proposed in different works by different authors but has been formalized as we know it in 1977 \citep{dempster1977maximum}. For more details, we refer to \citep{lindsay1995mixture, bilmes1998gentle}. EM is an iterative algorithm which is guaranteed to converge to a (local) optimum of the likelihood, and it is widely used to solve the ML or the MAP estimation problems. This algorithm has a striking number of applications and has been successfully used for medical imaging \citep{balafar2010review}, image restoration \citep{lagendijk1990identification}, problems in computational biology \citep{fan2010algorithm}, and so on. 

One of the most important applications of the EM algorithm is for fitting mixture models in machine learning \citep{murphy2012machine}. Most of the mixture models use a base distribution that belongs to the exponential family: Poisson \citep{church1995poisson}, Binomial,  Multinomial, log-normal \citep{dexter1972packing}, exponential \citep{ghitany1994exponential}, Dirichlet multinomial \citep{yin2014dirichlet}, and others. EM is also used to fit mixtures of experts, mixtures of the student T distribution (which does not belong to  the exponential family, and can be fitted with EM using \citep{liu1995ml}), factor analysis, probit regression, and learning Hidden Markov Models \citep{murphy2012machine}. The estimator of the parameter computed by EM has many relevant properties. For instance, it is asymptotically efficient: no other estimator can achieve asymptotically smaller variance in function of the number of points \cite{moitra2018algorithmic}. 
	
In this work, we introduce Quantum Expectation-Maximization (QEM), a quantum algorithm for fitting mixture models. We detail its usage in the context of Gaussian Mixture Models, and we extend the result to other distributions in the exponential family. It is straightforward to use the ML estimates to compute the MAP estimate of a mixture model. Our main result can be stated as:

	\begin{result*}[Quantum Expectation-Maximization] (see Theorem \ref{th:qgmm})
		Given quantum access to a GMM and a dataset $V \in \mathbb{R}^{n \times d}$ like in Definition \ref{def:quantumaccess}, for parameters $\delta_\theta, \delta_\mu > 0$,   Quantum Expectation-Maximization (QEM) fits a Maximum Likelihood (or a Maximum A Posteriori) estimate of a Gaussian Mixture Model with $k$ components, in running time per iteration which is dominated by:
   \begin{equation}\widetilde{O}\left( \frac{d^2k^{4.5}\eta^3\kappa(V)\kappa(\Sigma)\mu(\Sigma)}{\delta_\mu^3} \right),\end{equation}
where $\Sigma$ is a covariance matrix of a Gaussian distribution of one of the mixture , $\eta$ is a parameter of the dataset related to the maximum norm of the vectors, $\delta_\theta, \delta_\mu$ are error parameters in the QEM algorithm, $\mu(\Sigma)$ ( which is always $\leq \sqrt{d}$) is a factor appearing in quantum linear algebra and $\kappa$ is the condition number of a matrix. 
\end{result*}	
\noindent

This algorithm assumes quantum access to the dataset (see Definition \ref{def:quantumaccess}). This assumption can be satisfied with a preprocessing that can be carried on in time $\widetilde O(V)$ (i.e. upon reception of the dataset ).
It is typical for quantum algorithms to depend on a set of parameters, and they can achieve speedups with respect to classical algorithms when these parameters are within a certain range. Importantly, the value of these parameters can be estimated (as we do in this work for the VoxForge dataset ) and we expect it not to be too high on real datasets. Most of these parameters can be upper bounded, as described in the Experiment section. Remark that we have formal bounds for the condition number as $\kappa(V)\approx 1/\min( \{\theta_1, \cdots,  \theta_k\}\cup \{ d_{st}(\mathcal{N}(\mu_i, \Sigma_i), \mathcal{N}(\mu_j, \Sigma_j)) | i\neq j \in [k] \})$, where $d_{st}$ is the statistical distance between two Gaussian distributions. \cite{kalai2012disentangling}. Here we only kept the term in the running time that dominates for the range of parameters of interest, while in Theorem \ref{th:qgmm} we state explicitly the running times of each step of the algorithm. As in the classical case, QEM algorithm runs for a certain number of iterations until a stopping condition is met (defined by a parameter $\epsilon_\tau > 0$) which implies convergence to a (local) optimum. We remark that in the above result we performed tomography enough times to get an $\ell_2$ guarantee in the approximation, nevertheless, a lesser guarantee may be enough, for example using $\ell_\infty$ tomography (Theorem \ref{thm:tom1}), which can potentially remove the term $d^2$ from the running time.

Let's have a first high-level comparison of this result with the standard classical algorithms. The runtime of a single iteration in the standard implementation of the EM algorithm is  $O(knd^{2})$ \citep{scikit-learn,murphy2012machine}. The advantage of the quantum algorithm is an exponential improvement with respect to the number of elements in the training set, albeit with a worsening on other parameters. Note that we expect the number of iterations of the quantum algorithm to be similar to the number of iteration of the classical case, as the convergence rate is not expected to change, and this belief is corroborated by experimental evidence for the simpler case of univariate Gaussians of k-means \citep{kerenidis2019q}. In this work, we tested our algorithm on a non-trivial dataset for the problem of speaker recognition on the VoxForge dataset \cite{voxforge}. The experiment aimed to gauge the range of the parameter that affects the runtime, and test if the error introduced in the quantum procedures still allow the models to be useful. From the experiments reported in Section \ref{expevidence}, we believe that datasets where the number of samples is very large might be processed faster on a quantum computer. 
One should expect that some of the parameters of the quantum algorithm can be improved, especially the dependence on the condition numbers and the errors, which can  extend the number of datasets where QEM can offer a computational advantage. We also haven't optimized the parameters that govern the runtime in order to have bigger speedups, but we believe that hyperparameter tuning techniques, or a simple grid search, can improve the quantum algorithm even further. ML may not always be the best way to estimate the parameters of a model. For instance, in high-dimensional spaces, it is pretty common for ML estimates to overfit. MAP estimates inject into a ML model some external information, perhaps from domain experts. This usually avoids overfitting by having a kind of regularization effect on the model. For more information on how to use QEM for a MAP estimate we refer to \cite{murphy2012machine} and the Appendix.

This work is organized as such. First, we review previous efforts in quantum algorithms for unsupervised classification and classical algorithms for GMM. Then we present the classical EM algorithm for GMM. Then, in Section \ref{gmmsection} we describe QEM and its theoretical analysis, and we show how to use it to fit a GMM and other mixture models. We conclude by giving some experimental evidence on the expected runtime on real data.

\subsection{Related work}
Many classical algorithms for GMM exist. Already in 1895, Karl Pearson fitted manually a GMM of 2 features using the methods of moments \cite{pearson1895x}. Without resorting to heuristics - like the EM algorithm - other procedures with provable guarantees exist. For example, in \cite{dasgupta1999learning} they assume only one shared covariance matrix among mixtures, but they have a polynomial dependence on the number of elements in the training set, and in \cite{kannan2005spectral} they developed the spectral projection technique. Formal learnability of Gaussian mixtures has been studied \cite{kalai2012disentangling, moitra2010settling}. 
While these important results provide provable algorithms for the case when the Gaussians are well-separated, we think that the heuristic-based approaches, like the classical and the quantum EM, should not be directly compared to this class of algorithms.

In the quantum setting, a number of algorithms have been proposed in the context of unsupervised learning \cite{aimeur2013quantum, Lloyd2013a, otterbach2017unsupervised, kerenidis2019q}. Recently, classical machine learning algorithms were obtained by ``dequantizing'' quantum machine learning algorithms \cite{tang2018quantum, GLT18, tang2018quantum2, gilyen2018quantum, chia2018quantum}. This class of algorithms is of high theoretical interest, as runtime is poly-logarithmic in the dimensions of the dataset. However, the high polynomial dependence on the Frobenius norm, the error, and the condition number, makes this class of algorithms still impractical for interesting datasets, as shown experimentally \cite{arrazola2019quantum}. We believe there can be a ``dequantized'' version of QEM, but it seems rather unlikely that this algorithm will be more efficient than QEM or the classical EM algorithm.  

As the classical EM for GMM can be seen as a generalization of k-means, our work is a generalization of the $q$-means algorithm in \cite{kerenidis2019q}. Independently and simultaneously, Miyahara, Aihara, and Lechner also extended the $q$-means algorithm and applied it to fit a Gaussian Mixture Models \cite{miyahara2019expectation}. The main difference with this work is that the update step in \cite{miyahara2019expectation} is performed using a hard-clustering approach (as in the $k$-means algorithm): for updating the centroids and the covariance matrices of a cluster $j$, only the data points for which cluster $j$ is \emph{nearest} are taken into account. In our work, as in the classical EM algorithm, we use the soft clustering approach: that is, for updating the centroid and the covariance matrices of cluster $j$, all the data points \emph{weighted by their responsibility} (Defined in Eq. \ref{responsibility}) for cluster $j$ are taken into account. Both approaches have merits and can offer advantages \cite{kearns1998information}, albeit is more adherent to the original EM algorithm.

\section{Preliminaries}

Quantum computing provides a new way to encode information and perform algorithms. The basic carrier of quantum information is the \textit{qubit}, which can be in a superposition of two states $\ket{0}$ and $\ket{1}$ at the same time. More formally, a quantum bit can be written as: 
$ \ket{x} = \alpha \ket{0} + \beta\ket{1}$
and it corresponds to a unit vector in the Hilbert space $ \mathcal{H}_2=\text{span}\{\ket{0},\ket{1}\}$ with $\alpha,\beta \in \mathbb{C}$ and $|\alpha|^2+|\beta|^2=1$. An $n$-qubit state corresponds to a unit vector in $\mathcal{H}_n=\otimes_{i\in[n]} \mathcal{H}_2\sim\mathbb{C}^{2^n}$, there $\otimes$ is the tensor product. Denoting by $\{\ket{i}\}$ the standard basis of $\mathcal{H}_n$, an $n$-qubit state can be written as: $ \ket{x} = \sum_{i=0}^{2^n-1} \alpha_i \ket{i} $ with $\alpha_i \in \mathbb{C}$ and $\sum_i |\alpha_i|^2=1$. Quantum states evolve through unitary matrices, which are norm-preserving, and thus can be used as suitable mathematical description of pure quantum evolutions $U\ket{\psi}\mapsto\ket{\psi'}$. A matrix $U$ is said to be unitary if $UU^\dagger = U^\dagger U = I$. 
A quantum state $\ket{x}$ can be measured, and the probability that a measurement on $\ket{x}$ gives outcome $i$ is $|\alpha_i|^2$. The quantum state corresponding to a vector $v\in\mathbb{R}^m$ is defined as $\ket{v}=\frac{1}{\norm{v}}\sum_{j\in[m]} v_j \ket{j}$. Note that to build $\ket{v}$ we need $\lceil \log m\rceil$ qubits. In the following, when we say with high probability we mean a value which is inverse polynomially close to 1. The value of $\mu(V)$ which often compares in the runtimes comes from the procedure we use in the proof. While its value for real dataset is discussed in the manuscript, for a theoretical analysis we refer to \cite{kerenidis2017quantumsquares}. We denote as $V_{\geq \tau}$ the matrix  $\sum_{i=0}^\ell \sigma_i u_i v_i^T$ where $\sigma_\ell$ is the smallest singular value which is greater than $ \tau$. The dataset is represented by a matrix $V \in \R^{n \times d}$, i.e. each row is a vector $v_i \in \R^{d}$ for $i \in [n]$ that represents a single data point. The cluster centers, called centroids, at time $t$ are stored in the matrix $C^t \in \R^{k \times d}$, such that the $j^{th}$ row $c_{j}^{t}$ for $j\in [k]$ represents the centroid of the cluster $\mathcal{C}_j^t$. The number of non-zero elements of $V$ is $nnz(V)$. Let $\kappa(V)$ be the condition number of $V$: the ratio between the biggest and the smallest (non-zero) singular value. We recommend Nielsen and Chuang \citep{NC02} for an introduction to quantum information.

We will use a definition from probability theory.

\begin{definition}[Exponential Family \citep{murphy2012machine}] A probability density function or probability mass function $p(v|\nu)$ for $v = (v_1, \cdots, v_m) \in \mathcal{V}^m$, where $\mathcal{V} \subseteq \mathbb{R}$, %
	$\nu \in \mathbb{R}^p$ is said to be in the exponential family if can be written as:
 $$p(v|\nu) := 	h(v)\exp \{ o(\nu)^TT(v) - A(\nu) \}$$
 where:   
 \begin{itemize}
		\item $\nu \in \mathbb{R}^p$ is called the \emph{canonical or natural} parameter of the family,
		\item $o(\nu)$ is a function of $\nu$  (which often is just the identity function),  
		\item $T(v)$  is the vector of sufficient statistics: a function that holds all the information the data $v$ holds with respect to the unknown parameters,
		\item $A(\nu)$ is the cumulant generating function, or log-partition function, which acts as a normalization factor,
		\item $h(v) > 0$ is the \emph{base measure} which is a non-informative prior and de-facto is scaling constant.
 \end{itemize}
 \end{definition}

 \begin{theorem}[Multivariate Mean Value Theorem \citep{rudin1964principles}]
    Let $U$ be an open set of $\R^d$. For a differentiable functions $f : U \mapsto \mathbb{R}$ it holds that  $ \forall x,y \in U, \: \exists c$ such that  $f(x) - f(y) = \nabla f(c) \cdot (x-y) $.
 \end{theorem}

 \begin{lemma}[Componentwise \emph{Softmax} function $\sigma_j(v)$ is Lipschitz continuous]\label{softmax}
    For $d>2$, let $\sigma_j : \R^d \mapsto (0,1)$ be the \emph{softmax} function defined as $\sigma_j(v) = \frac{e^{v_j}}{\sum_{l=1}^de^{v_l}}$ Then $\sigma_j$ is Lipschitz continuous, with $K \leq \sqrt{2}$.
 \end{lemma}
 \begin{proof}
    We need to find the $K$ such that for all $x,y \in \R^d$, we have that  $\norm{\sigma_j(y) - \sigma_j(x)} \leq K\norm{y-x}$. Observing that $\sigma_j$ is differentiable and that if we apply Cauchy-Schwarz to the statement of the Mean-Value-Theorem we derive that $ \forall x,y \in U, \: \exists c$ such that  $\norm{f(x) - f(y)} \leq \norm{\nabla f(c)}_F \norm{x-y} $.
    So to show Lipschitz continuity it is enough to select $K \leq\norm{\nabla \sigma_j}_F^{*} = \max_{c \in \R^d} \norm{\nabla \sigma_j(c)}$.

    The partial derivatives $\frac{d \sigma_j(v)}{d v_i}$ are $\sigma_j(v)(1-\sigma_j(v))$ if $i=j$ and $-\sigma_i(v)\sigma_j(v)$ otherwise. 
    So $\norm{\nabla \sigma_j}_F^2 = \sum_{i=1}^{d-1} (-\sigma(v)_i\sigma_j(v))^2 + \sigma_j(v)^2(1-\sigma_j(v))^2  \leq \sum_{i=1}^{d-1} \sigma(v)_i\sigma_j(v) + \sigma_j(v)(1-\sigma_j(v)) \leq \sigma_j(v) \sum_{i=0}^{d-1} \sigma_i(v) + 1 - \sigma_j(v) \leq 2\sigma_j(v) \leq 2$. 
    In our case we can deduce that:
    $\norm{\sigma_j(y) - \sigma_j(x)} \leq \sqrt{2} \norm{y-x} $ so $K\leq \sqrt{2}$.

 \end{proof}

     \begin{lemma}[Error in the responsibilities of the exponential family]\label{respsoftmaxed}
         Let $v_i \in \R^{n}$ be a vector, and let $\{ p(v_i | \nu_j)\}_{j=1}^k$ be a set of $k$ probability distributions in the exponential family, defined as $p(v_i | \nu_j):=h_j(v_i)exp\{o_j(\nu_j)^TT_j(v_i) - A_j(\nu_j)\}$. Then, if we have estimates 
         for each exponent with error $\epsilon$, then we can compute each $r_{ij}$ such that $|\overline{r_{ij}} - r_{ij}| \leq \sqrt{2k}\epsilon$ for $j \in [k]$.
     \end{lemma}
     
     \begin{proof}
         The proof follows from rewriting the responsibility of Equation \eqref{responsibility} as:
         
         \begin{equation}
         r_{ij}  := \frac{h_j(v_i)\exp \{ o_j(\nu_j)^TT(v_i) - A_j(\nu_j)  \  + \log \theta_j \}}{\sum\limits_{l=1}^k h_l(v_i)\exp \{ o_l(\nu_l)^TT(v_i) - A_l(\nu_l)  \  + \log \theta_l \} }
         \end{equation}  
         In this form, it is clear that the responsibilities can be seen a \emph{softmax} function, and we can use Theorem \ref{softmax} to bound the error in computing this value. 
         
         Let $T_i \in \mathbb{R}^k$ be the vector of the exponent, that is $t_{ij} = o_j(\nu_j)^TT(v_i) - A_j(\nu_j)  + \log \theta_j $. In an analogous way we define $\overline{T_i}$ the vector where each component is the estimate with error $\epsilon$. The error in the responsibility is defined as $|r_{ij} - \overline{r_{ij}}| = |\sigma_j(T_i) - \sigma_j(\overline{T_i}) | $. Because the function $\sigma_j$ is Lipschitz continuous, as we proved in Theorem \ref{softmax} with a Lipschitz constant $K\leq \sqrt{2}$, we have that, $|  \sigma_j(T_i) - \sigma_j(\overline{T_i})  | \leq \sqrt{2} \norm{T_i - \overline{T_i}}$. The result follows as $\norm{T_i - \overline{T_i}} < \sqrt{k}\epsilon $.
     \end{proof}

To prove our results, we are going to use the quantum procedures listed hereafter. 

\begin{claim}\label{kereclaim}\citep{kerenidis2017quantumsquares}
   Let $\theta$ be the angle between vectors $x,y$, and assume that $\theta < \pi/2$. Then, $\norm{x-y} \leq \epsilon$ implies $\norm{\ket{x} - \ket{y}} \leq \frac{\sqrt{2}\epsilon}{\norm{x}}$. Where $\ket{x}$ and $\ket{y}$ are two unit vectors in $\ell_2$ norm.
\end{claim}

\noindent We will also use Claim 4.5 from \citep{kerenidis2019q}. 
\begin{claim}\label{quattrocinque}\citep{kerenidis2019q}
   Let $\epsilon_b$ be the error we commit in estimating $\ket{c}$ such that $\norm{ \ket{c} - \ket{\overline{c}}} < \epsilon_b$, and $\epsilon_a$ the error we commit in the estimating the norms,  $|\norm{c} - \overline{\norm{c}}| \leq \epsilon_a \norm{c} $. Then $\norm{\overline{c} - c} \leq \sqrt{\eta} (\epsilon_a +  \epsilon_b)$. 
\end{claim}

\begin{theorem}[Amplitude estimation and amplification \citep{BHMT00}]
   If there is unitary operator $U$ such that $U\ket{0}^{l}= \ket{\phi} = \sin(\theta) \ket{x, 0} + \cos(\theta) \ket{G, 0^\bot}$ then  $\sin^{2}(\theta)$ can be estimated to multiplicative error $\eta$ in time $O(\frac{T(U)}{\eta \sin(\theta)})$ and 
   $\ket{x}$ can be generated in expected time $O(\frac{T(U)}{\sin (\theta)})$.
\end{theorem}

We also need some state preparation procedures. These subroutines are needed for encoding vectors in $v_{i} \in \R^{d}$ into quantum states $\ket{v_{i}}$. An efficient state preparation procedure is provided by the QRAM data structures. We stress the fact that our results continue to hold, no matter how the efficient quantum loading of the data is provided. For instance, the data can be accessed through a QRAM, through a block encoding, or when the data can be produced by quantum circuits.

\begin{theorem}[QRAM data structure \citep{KP16}]\label{qram}
   Let $V \in \mathbb{R}^{n \times d}$, there is a data structure to store the rows of $V$ such that, 
   \begin{enumerate}
      \item The time to insert, update or delete a single entry $v_{ij}$ is $O(\log^{2}(n))$. 	
      \item A quantum algorithm with access to the data structure can perform the following unitaries in time $T=O(\log^{2}N)$. 
      \begin{enumerate} 
         \item $\ket{i}\ket{0} \to \ket{i}\ket{v_{i}} $ for $i \in [n]$. 
         \item $\ket{0} \to \sum_{i \in [n]} \norm{v_{i}}\ket{i}$. 
      \end{enumerate}
   \end{enumerate}
\end{theorem}

In our algorithm we will also use subroutines for quantum linear algebra. 
For a symmetric matrix $M \in \R^{d\times d}$ with spectral norm $\norm{M}=1$, the running time of these algorithms depends linearly on the condition number $\kappa(M)$ of the matrix, that can be replaced by $\kappa_\tau(M)$, a condition threshold where we keep only the singular values bigger than $\tau$, and the parameter $\mu(M)$, a matrix dependent parameter defined as 
$$\mu(M)= \min_{p \in P} (\norm{M}_F, \sqrt{s_{2p}(M)s_{2(1-p)}(M^{T})}),$$ 
for $s_{p}(M) := \max_{i \in [n]} \norm{m_i}_p^p$ where $\norm{m_i}_p$ is the $\ell_p$ norm of the i-th row of $M$, and $P$ is a finite set of size $O(1) \in [0,1]$.
Note that $\mu(M) \leq \norm{M}_{F} \leq \sqrt{d}$ as we have assumed that $\norm{M}=1$. The running time also depends logarithmically on the relative error $\epsilon$ of the final outcome state.  \citep{CGJ18, GLSW18}. 

\begin{theorem}[Quantum linear algebra \citep{CGJ18, GLSW18} ]\label{qla}  Let $M \in \mathbb{R}^{d \times d}$ such that $\norm{M}_2 =1$ and $x \in \mathbb{R}^d$. Let $\epsilon,\delta>0$. If we have quantum access to $M$ and the time to prepare $\ket{x}$ is $T_{x}$, then there exist quantum algorithms that with probability at least $1-1/poly(d)$ return
   a state $\ket{z}$ such that $\norm{ \ket{z} - \ket{Mx}} \leq \epsilon$ in time $\widetilde{O}((\kappa(M)\mu(M) + T_{x} \kappa(M)) \log(1/\epsilon))$.      
\end{theorem}

\begin{theorem}[Quantum linear algebra for matrix products \citep{CGJ18} ]\label{qlap}  Let $M_1, M_2 \in \mathbb{R}^{d \times d}$ such that $\norm{M}_1 = \norm{M}_2 =1$ and $x \in \mathbb{R}^d$, and a vector $x \in \mathbb{R}^d$ for which we have quantum access. Let $\epsilon>0$. Then there exist quantum algorithms that with probability at least $1-1/poly(d)$ returns a state $\ket{z}$ such that $\norm{ \ket{z} - \ket{Mx}} \leq \epsilon$ in time $\widetilde{O}((\kappa(M)( \mu(M_1)T_{M_1}+ \mu(M_2)T_{M_2} ) ) \log(1/\epsilon))$, where $T_{M_1},T_{M_2}$ is the time needed to index the rows of $M_1$ and $M_2$.
\end{theorem}

The linear algebra procedures above can also be applied to any rectangular matrix $V \in \mathbb{R}^{n \times d}$ by considering instead the symmetric matrix $ \overline{V} = \left ( \begin{matrix}
0  &V \\ 
V^{T} &0 \\
\end{matrix}  \right )$.

The final component needed for the QEM algorithm is an algorithm for vector state tomography that will be used to recover classical information from the quantum states corresponding to the new centroids in each step. We report here  two kind of vector state tomography. The $\ell_2$ tomography has stronger guarantees on the output, while the $\ell_\infty$ is faster. 

\begin{theorem}[$\ell_{\infty}$ Vector state tomography]\cite{jonas} \label{thm:tom1}
	Given access to unitary $U$ such that $U\ket{0} = \ket{x}$ and its controlled version in time $T(U)$, there is a tomography algorithm with time complexity $O(T(U) \frac{ \log d  }{\delta^{2}})$ that produces unit vector $\widetilde{X} \in \R^{d}$ such that $\norm{\widetilde{X}  - x }_{\infty} \leq \delta$ with probability at least $(1-1/poly(d))$. 
\end{theorem}

\begin{theorem}  [Vector state tomography \citep{KP18}] \label{thm:tom2} Given access to unitary $U$ such that $U\ket{0} = \ket{x}$ and its controlled version in time $T(U)$, there is a tomography algorithm with time complexity $O(T(U) \frac{ d \log d  }{\epsilon^{2}})$ that produces unit vector $\widetilde{x} \in \R^{d}$ such that $\norm{\widetilde{x}  - x }_{2} \leq \epsilon$ with probability at least $(1-1/poly(d))$. 
\end{theorem} 

\begin{lemma}[Distance / Inner Products Estimation \citep{kerenidis2019q, WKS14, LMR13}]\label{th:innerproductestimation} 
   Assume for a data matrix $V \in \mathbb{R}^{n \times d}$ and a centroid matrix $C \in \mathbb{R}^{k \times d}$ that the following unitaries $
   \ket{i}\ket{0} \mapsto \ket{i}\ket{v_i}, $ and $\ket{j}\ket{0} \mapsto \ket{j}\ket{c_j}
   $ can be performed in time $T$ and the norms of the vectors are known. For any $\Delta > 0$ and $\epsilon>0$, there exists a quantum algorithm that  computes 

$$   \ket{i}\ket{j}\ket{0}  \mapsto  \ket{i}\ket{j}\ket{\overline{d^2(v_i,c_j)}}, $$
where$ |\overline{d^{2}(v_i,c_j)}-d^{2}(v_i,c_j)| \leqslant  \epsilon$ with probability at least$ 1-2\Delta,$
or $$ \ket{i}\ket{j}\ket{0}  \mapsto \ket{i}\ket{j}\ket{\overline{(v_i,c_j)}},  $$
      where $ |\overline{(v_i,c_j)}-(v_i,c_j)| \leqslant  \epsilon $ with probability at least $ 1-2\Delta$
      in time $\widetilde{O}\left(\frac{ \norm{v_i}\norm{c_j} T \log(1/\Delta)}{ \epsilon}\right)$. 
\end{lemma}

\begin{lemma}[Estimation of quadratic forms]\label{lemma:quadratic forms}
    Assume to have quantum access to a symmetric positive definite matrix $A \in \mathbb{R}^{n \times n}$ such that $\norm{A}\leq 1$, and quantum access to a matrix $V \in \mathbb{R}^{n \times d}$. For $\epsilon,\delta > 0$, there is a quantum algorithm that succeeds with probability at least $1-2\delta$, and perform the mapping $\ket{i} \mapsto \ket{i}\ket{\overline{s_i}}$, for $|s_i - \overline{s_i}| \leq \epsilon$, where $s_i$ is either:
    \begin{itemize}
    \item $(\ket{v_i},A\ket{v_i})$ in time $O(\frac{\mu(A)}{\epsilon})$
    \item $(\ket{v_i},A^{-1}\ket{v_i})$ in time $O(\frac{\mu(A)\kappa(A)}{\epsilon})$
    \end{itemize}
    The algorithm can return an estimate of $\overline{(v_i, Av_i)}$ such that $\overline{(v_i, Av_i)} - (v_i, Av_i) \leq \epsilon$ using quantum access to the norm of the rows of $V$ by increasing the runtime by a factor of $\norm{v_i}^2$. 
    \end{lemma}
    
    \begin{proof}
    Let's analyze first the case where we want to compute the quadratic form with $A$, and after the case for $A^{-1}$. We can use Theorem \ref{qla} to perform the following mapping:
    
\begin{align}\label{eq:quadatic A}
        \ket{i}\ket{v_i} & \mapsto  \ket{i} \Big(\norm{Av_i}\ket{Av_i,0} +  \sqrt{1-\gamma^2}\ket{G,1} \Big)\\
        &= \ket{i}\ket{\psi_i}
\end{align}

For the case where we want to compute the quadratic form for $A^{-1}$, we define the constant $C = O(1/\kappa(A))$, and define $\ket{\psi_i}$ as the result of the mapping:
     \begin{align}\label{eq:quadatic A minus 1}
      \ket{i}\ket{v_i}\mapsto  \ket{i} \Big( C\norm{A^{-1}v_i}\ket{A^{-1}v_i,0} + \sqrt{1-\gamma^2}\ket{G,1} \Big)
 \end{align}

We also define a second register $\ket{i}\ket{v_i,0}$. Using controlled operations, we can create the state:
    
\begin{align}\label{eq:measure}
    \frac{1}{2}\ket{i}\left(\ket{0}(\ket{v_i,0}+\ket{\psi_i}) + \ket{1}(\ket{v_i,0}-\ket{\psi_i})   \right)
\end{align}
    It is simple to check that, for a given register $\ket{i}$, the probability of measuring $0$ is:
    $$p_i(0) = \frac{1+\norm{Av_i}\braket{Av_i|v_i}}{2}$$
    
For the case of $A^{-1}$, the probability of measuring $0$ in state of Equation \ref{eq:measure} is 
    $$p_i(0) = \frac{1+C\norm{A^{-1}v_i}\braket{A^{-1}v_i|v_i}}{2} $$
    
    For both cases, we are left with the task of coherently estimate the measurement probability in a quantum register and boost the success probability of this procedure. The unitaries that create the states in Equation \ref{eq:quadatic A} and \ref{eq:quadatic A minus 1} plugged into Equation \ref{eq:measure} (i.e before a measurement on the ancilla qubit) describe a mapping: $U_1:\ket{i} \mapsto \frac{1}{2}\ket{i} \left(\sqrt{p_i(0)}\ket{y_i,0} + \sqrt{1-p_i(0)}\ket{G_i,1} \right)$. We follow similarly the steps of the proof for inner product estimation in \cite{kerenidis2019q}. We use amplitude estimation in its original formulation, i.e. Theorem 12 of \cite{BHMT00}, along with median evaluation Lemma 8 of \cite{wiebe2014quantum}, called with failure probability $\delta$. Amplitude estimation gives a unitary that perform
\begin{equation}
            U_2 \ket{i} \mapsto \frac{1}{2}\ket{i} \left(\sqrt{\alpha}\ket{p_i(0),y_i,0} + \sqrt{1-\alpha}\ket{G'_i,1} \right)
\end{equation}
            and estimate $p_i(0)$ such that $|p_i(0) - \overline{p_i(0)}| < \epsilon$ for the case of $v_i^TAv_i$ and we choose a precision $\epsilon/C$ for the case of $v_i^TA^{-1}v_i$ to get the same accuracy. The version of amplitude estimation Theorem we used fails with probability $\leq \frac{8}{\pi^2}$, and thus, is suitable to be used in the median evaluation Lemma, which boost the success probability to $1-\delta$. The runtime of this procedure is given by combining the runtime of creating state $\ket{\psi_i}$, amplitude estimation, and the median Lemma. Since the error in the matrix multiplication step is negligible, and assuming quantum access to the vectors is polylogarithmic, the final runtime is of $O(\log(1/\delta)\mu(A)\log(1/\epsilon_2) / \epsilon)$, with an additional factor $\kappa(A)$ for the case of quadratic form of $A^{-1}$.				
            Note that if we want to estimate a quadratic form of two unnormalized vectors, we can just multiply this result by their norms. Note also that the absolute error $\epsilon$ now becomes relative w.r.t the norms, i.e. $\epsilon \norm{v_i}^2$.  If we want to obtain an absolute error $\epsilon'$, as in the case with normalized unit vectors, we have to run amplitude estimation with precision $\epsilon'=O(\epsilon/\norm{v_i}^2)$.
            To conclude, this subroutines succeeds with probability $1-\gamma$ and requires time $O(\frac{\mu(A) \log (1/\gamma) \log (1/\errmult)}{\epsilon_1})$, with an additional factor of $\kappa(A)$ if we were to consider the quadratic form for $A^{-1}$ and an additional factor of $\norm{v_i}^2$ if we were to consider the non-normalized vectors $v_i$. This concludes the proof of the Lemma. 
\end{proof}

\section{EM and Gaussian Mixture Models}

    GMM are probably the most used mixture model used to solve unsupervised classification problems. In unsupervised settings, we are given a training set of unlabeled vectors $v_1 \dots v_n \in \mathbb{R}^d$ which we represent as rows of a matrix $V \in \mathbb{R}^{n \times d}$. 
Let $y_i \in [k]$ one of the $k$ possible labels for a point $v_i$. We posit that the joint probability distribution of the data $p(v_i, y_i)=p(v_i | y_i)p(y_i)$, is defined as follow: $y_i \sim \text{Multinomial}(\theta)$ for $\theta \in \mathbb{R}^k$,  and $p(v_i|y_i = j) \sim \mathcal{N}(\mu_j, \Sigma_j) $. The $\theta_j$ such that $\sum_j \theta_j=1$ are called \emph{mixing weights}, i.e. the probabilities that $y_i = j$, and $ \mathcal{N}(\mu_j, \Sigma_j) $ is the Gaussian distribution centered in $\mu_j \in \mathbb{R}^d$ with covariance matrix $\Sigma_j \in \mathbb{R}^{d \times d}$  \citep{ng2012cs229}.  %
We use the letter $\phi$ to represent our \emph{base distribution}, which in this case is the probability density function of a multivariate Gaussian distribution $\mathcal{N}(\mu, \Sigma)$. Using this formulation, a GMM is expressed as: $		p(v) = \sum_{j=1}^k \theta_j \phi(v; \mu_j, \Sigma_j) $.  Fitting a GMM to a dataset reduces to finding an assignment for the parameters of the model
$\gamma = (\theta, \bm \mu, \bm \Sigma) = (\theta, \mu_1, \cdots, \mu_k, \Sigma_1, \cdots, \Sigma_k) $
that best maximize the log-likelihood for a given dataset. Note that the algorithm used to fit GMM can return a local minimum which might be different than $\gamma^*$: the model that represents the global optimum of the likelihood function. For a given $\gamma$, the probability for an observation $v_i$ to be assigned to the component $j$ is given by: 
\begin{equation}\label{responsibility}
r_{ij} = \frac{\theta_j \phi(v_i; \mu_j, \Sigma_j )}{\sum_{l=1}^k \theta_l \phi(v_i; \mu_l, \Sigma_l).}\end{equation}
This value is called \emph{responsibility}, and corresponds to the posterior probability of the sample $i$ being assigned label $j$ by the current model. As anticipated, to find the best parameters of our generative model, we
	maximize the log-likelihood of the data. Alas, it is seldom possible to solve maximum likelihood estimation analytically
	(i.e. by finding the zeroes of the derivatives of the log-like function) for mixture models like the GMM. To complicate things, the likelihood function for GMM is not convex, and thus we might find some local minima \citep{friedman2001elements}.

	EM is an iterative algorithm that solves numerically the optimization problems linked to ML and MAP estimations.
    The classical EM algorithm works as follows: in the expectation step all the responsibilities are calculated, and we estimate the missing variables $y_i$'s given the current guess of the parameters $(\theta, \bm \mu, \bm \Sigma )$ of the model. 
	Then, in the maximization step, we use the estimate of the latent variables 
	obtained in the expectation step to update the estimate of the parameters: $\gamma^{t+1} = (\theta^{t+1}, \bm \mu^{t+1}, \bm \Sigma^{t+1} )$. 
	While in the Expectation step we calculate a lower bound on the likelihood, in the maximization step we maximize it. Since at each iteration the likelihood can only increase, the algorithm is guaranteed to converge, albeit possibly to a local optimum (see \citep{friedman2001elements} for the proof). 
    The stopping criterion is usually a threshold on the increment of the log-likelihood: if it changes less than a $\epsilon_\tau$ between two iterations, then the algorithm stops. As the value of the log-likelihood depends on the amount of data points in the training sets, it is often preferable to adopt a scale-free stopping criterion, which does not depend on the number of samples. For instance, in scikit-learn \citep{scikit-learn} the stopping criterion is given by a tolerance on the average increment of the log-probability, chosen to be some $10^{-3}$. More precisely, the stopping criterion in the quantum and classical algorithm is 
    $ | \mathbb{E}[ \log p(v_i;\gamma^{t})] - \mathbb{E}[\log p(v_i;\gamma^{t+1})] | < \epsilon_\tau.$

    \section{Quantum EM for GMM}\label{gmmsection}
    In this section, we present a quantum EM algorithm to fit a GMM. The algorithm can also be adapted fit other mixtures models where the probability distributions belong to the exponential family. As the GMM is both intuitive and one of the most widely used mixture models, our results are presented for the GMM case. As in the classical algorithm, we use some subroutines to compute the responsibilities and update our current guess of the parameters which resemble the E and the M step of the classical algorithm. While the classical algorithm has clearly two separate steps for Expectation and maximization, the quantum algorithm uses a subroutine to compute the responsibilities inside the step that performs the maximization. During the quantum M step, the algorithm updates the model by creating quantum states corresponding to parameters $\gamma^{t+1}$ and then recovering classical estimates for these parameters using quantum tomography or amplitude amplification. In order for the subroutines to be efficient, we build quantum access (as in Definition \ref{def:quantumaccess}) to the current estimate of the model, and we update it at each maximization step.

    \paragraph{Dataset assumptions in GMM} \label{assume} 
    In the remainig of the paper we make an assumption on the dataset, namely that all elements of the mixture contribute proportionally to the total responsibility:
    \begin{equation}\label{assump} \frac{\sum_{i=1}^n r_{ij}}{\sum_{i=1}^n r_{il}} = \Theta(1) \quad \forall j,l \in [k] \end{equation}
     This is equivalent to assuming that $\theta_j/\theta_l = \Theta(1) \quad \forall j,l \in [k]$.  
     The algortihm we propose can of course be used even in cases where this assumption does not hold. In this case, the running time will include a factor as in Eq. \ref{assump} which for simplicity we have taken as constant in what follows. Note that classical algorithms would also find it difficult to fit the data in certain cases, for example when some of the clusters are very small. In fact, it is known (and not surprising) that if the statistical distance between the probability density function of two different Gaussian distributions is smaller than $1/2$, then we can not tell for a point $v$ from which Gaussian distribution it belongs to, even if we knew the parameters \cite{moitra2018algorithmic}. Only for convenience in the analysis, we also assume the dataset as being normalized such that the shortest vector has norm $1$ and define $\eta:=max_i \norm{v_i}^2$ to be the maximum norm squared of a vector in the dataset.

    \paragraph{An approximate version of GMM}
  Here we define an approximate version of GMM, that we fit with QEM algorithm. The difference between this formalization and the original GMM is simple. Here we make explicit in the model the \emph{approximation error} introduced during the training algorithm. 
    
  \begin{definition}[Approximate GMM]\label{def:approxgmm}
    Let $\gamma^{t}=(\theta^{t}, \bm \mu^{t}, \bm \Sigma^{t}) = (\theta^{t}, \mu^{t}_1 \cdots \mu^{t}_k, \Sigma^{t}_1 \cdots \Sigma^{t}_k) $   a model fitted by the standard EM algorithm from $\gamma^{0}$ an initial guess of the parameters, i.e. $\gamma^{t}$ is the error-free model that standard EM would have returned after $t$ iterations. Starting from the same choice of initial parameters $\gamma^{0}$, fitting a GMM with the QEM algorithm with $\Delta=(\delta_\theta, \delta_\mu)$ means returning a model $\overline{\gamma}^{t} = (\overline{\theta}^{t}, \overline{\bm \mu}^{t}, \overline{\bm \Sigma}^{t})$ such that:
    \begin{itemize}
        \item  $\norm{\overline{\theta}^{t} - \theta^{t}} < \delta_\theta$,	
        \item  $\norm{\overline{\mu_j}^{t} - \mu_j^{t}} < \delta_\mu$\: for all $j \in [k]$, 
	\item $\norm{\overline{\Sigma_j}^{t} - \Sigma_j^{t}}  \leq \delta_\mu\sqrt{\eta}$ \: for all $j \in [k]$.
     \end{itemize}
    \end{definition}

    \paragraph{Quantum access to the mixture model} 

    Here we explain how to load into a quantum computer a GMM and a dataset represented by a matrix $V$. This is needed for a quantum computer to be able to work with a machine learning model. The definition of quantum access to other kind of models is analogous. 

        \begin{definition}[Quantum access to a GMM]\label{def:quantumaccess}
            We say that we have quantum access to a GMM of a dataset $V \in \R^{n\times d}$ and model parameters $\theta_j \in \mathbb{R}$, $\mu_{j} \in \R^{d}, \Sigma_{j} \in \R^{d\times d}$ for all $j\in [k]$  if we can perform in time  $O(\text{polylog}(d))$ the following mappings: 
            \begin{itemize}
                \item $\ket{j} \ket{0} \mapsto \ket{j}\ket{\mu_j}$,
                \item $\ket{j} \ket{i} \ket{0} \mapsto \ket{j} \ket{i}\ket{\sigma_i^j}$ for $i \in [d]$ where $\sigma_i^{j} $ is the $i$-th rows of $\Sigma_j \in \mathbb{R}^{d \times d}$, 
                \item $\ket{i}\ket{0}\mapsto\ket{i}\ket{v_i}$ for all $i \in [n]$,
                \item $\ket{i}\ket{0}\ket{0}\mapsto\ket{i}\ket{\text{vec}[v_iv_i^T]} = \ket{i}\ket{v_i}\ket{v_i}$ for $i \in [n]$,  
                \item $\ket{j} \ket{0}\mapsto \ket{j}\ket{\theta_j}$.
    \end{itemize}
        \end{definition}
      
For instance, one may use a QRAM data structure as in \cite{kerenidis2016recommendation,kerenidis2017quantumsquares} and Definition \ref{qram}. Quantum access to a matrix can be reduced to being able to perform the unitary mapping $\ket{i}\ket{j}\ket{0} \xrightarrow{} \ket{i}\ket{j}\ket{a_{ij}}$, in other words being able to map the indices $(i,j)$ to the element of the matrix $a_{ij}$, which is the usual quantum oracle model.

    \subsection{QEM algorithm and analysis}
        
        \begin{algorithm}[ht]
            \caption{QEM for GMM  } \label{qGMMEM}
            \begin{algorithmic}[1]
                
                \REQUIRE Quantum access to a GMM model, precision parameters $\delta_\theta, \delta_\mu,$ and threshold $\epsilon_\tau$. 
                \ENSURE  A GMM $\overline{\gamma}^t$ that maximizes locally the likelihood $\ell(\gamma;V)$, up to tolerance $\epsilon_\tau$. 
                \vspace{10pt} 
                \STATE Use a heuristic (Supp. Material) to determine the initial guess  $\gamma^0=(\theta^0, \bm \mu^0, \bm \Sigma^0)$, and build quantum access as in Definition \ref{def:quantumaccess} those parameters. 
                \STATE Use Lemma \ref{lemma:det} to estimate the log determinant of the matrices $\{ \Sigma_j^0 \}_{j=1}^k$.
                \STATE t=0
                \REPEAT 
                 \STATE \textbf{Step 1:} Get an estimate of $\theta^{t+1}$ using Lemma \ref{lemma:theta} such that 
                 $\norm{\overline{\theta}^{t+1} - \theta^{t+1}} \leq\delta_\theta. $ %
                \STATE \textbf{Step 2:} Get an estimate $\{ \overline{\mu_j}^{t+1} \}_{j=1}^k$ by using Lemma \ref{lemma:centroids} to estimate  each $\norm{\mu_j^{t+1}} $ and $\ket{\mu_j^{t+1}}$ such that 
                $\norm{\mu_j^{t+1} - \overline{\mu_j}^{t+1}} \leq \delta_\mu $. 
    
                \STATE \textbf{Step 3:} Get an estimate $\{ \overline{\Sigma_j}^{t+1} \}_{j=1}^k$ by using Lemma \ref{lemma:sigma} to estimate $\norm{\Sigma_j^{t+1}}_F$ and $\ket{\Sigma_j^{t+1}}$ such that 
                $ \norm{\Sigma_j^{t+1} - \overline{\Sigma_j}^{t+1} } \leq \delta_{\mu} \sqrt{\eta} $. 
                \STATE \textbf{Step 4:} Estimate $\mathbb{E}[ \overline{p(v_i;\gamma^{t+1})}]$  up to error $\epsilon_\tau/2$ using Theorem \ref{lemma:likelihoodestimation}.  
                \STATE \textbf{Step 5:} Build quantum access to $\gamma^{t+1}$, and use Lemma \ref{lemma:det} to estimate the determinants $\{ \overline{\log det(\Sigma_j^{t+1})} \}_{j=0}^k$.
                \STATE $t=t+1$
                \UNTIL 
                $$| \mathbb{E}[ \overline{p(v_i;\gamma^{t})}] - \mathbb{E}[\overline{p(v_i;\gamma^{t-1})}] | < \epsilon_\tau$$
                \STATE Return $\overline{\gamma}^{t}=(\overline{\theta}^t, \overline{\bm \mu}^t, \overline{\bm \Sigma}^t )$
            \end{algorithmic}
        \end{algorithm}

We describe the different steps of QEM as Algorithm \ref{qGMMEM}, and analyze its running time, by showing the runtime and the error for each of the steps of the algorithm. Quantum initialization strategies exists, and are described in the Supp. Material. In the statement of the following Lemmas, $\Sigma$ is defined as one of the covariance matrices of the mixture. It appears in the runtime of the algorithms because is introduced during the proofs, which can be found in the Supp. Material.

\subsubsection{Expectation}
        In this step of the algorithm we are just showing how to compute efficiently the responsibilities as a quantum state. First, we compute the responsibilities in a quantum register, and then we show how to put them as amplitudes of a quantum state. We start by a classical algorithm used to efficiently approximate the log-determinant of the covariance matrices of the data. At each iteration of QEM we need to approximate the log-determinant of the updated covariance matrices using Lemma \ref{lemma:det}.  We will see from the error analysis that in order to get an estimate of the GMM, we need to call Lemma \ref{lemma:det} with precision for which the runtime of Lemma \ref{lemma:det} gets subsumed by the running time of finding the updated covariance matrices through Lemma \ref{lemma:sigma}. Thus, we do not explicitly write the time to compute the log-determinant from now on in the running time of the algorithm and when we say that we update $\Sigma$ we include an update on the estimate of $\log(\det(\Sigma))$ as well. The method can be extended to fit other kind of mixture models with base distribution in exponential family, by replacing the function to compute the posterior probabilities $p(v_i|j)$ (Lemma \ref{lemma:responsibility}) by a different procedure for computing the responsibility.  This  technique can be used to bound the error of Lemma \ref{lemma:responsibility} for distributions in the exponential family using the smoothness properties of the softmax function. The log-determinant in the Gaussian case generalizes to the cumulant generating function $A(\nu)$ for the exponential families. A fast quantum algorithm for the estimation of the log-determinant has been put forward by \cite{boutsidis2017randomized}. Afterwards \cite{han2015large} improved the dependence upon the condition number by a factor of square root.

\begin{theorem}[\citep{han2015large}]\label{cdet}
    Let $M\in \mathbb{R}^{d \times d}$ be a positive definite matrix with eigenvalues in the interval $(\sigma_{min}, 1)$. Then for all $\delta \in (0,1)$ and $\epsilon >0$ there is a classical algorithm that outputs $\overline{\log (det(M))}$ such that $| \overline{\log (det(M))} - \log(det(M))| \leq 2\epsilon |\log(det(M))|$
    with probability at least $1-\delta$ in time
    $$ T_{Det, \epsilon} := O \left( \sqrt{\kappa(M)}\frac{\log(1/\epsilon)\log(1/\delta) }{\epsilon^2}nnz(M) \right).  $$
 \end{theorem}

As a consequence of the previous Theorem, we have the following:

 \begin{theorem}[Determinant estimation]\label{lemma:det}
    There is an algorithm that, given as input a matrix $\Sigma$ and a parameter $0 < \delta < 1$,  outputs an estimate $ \overline{\log(det(\Sigma) )}$ such that $| \overline{\log (det(\Sigma))} - \log(det(\Sigma))| \leq \epsilon$ with probability $1-\delta $ in time:
    $$T_{Det,\epsilon} = \widetilde{O} \left( \epsilon^{-2}\sqrt{\kappa(\Sigma)}\log(1/\delta)nnz(\Sigma)|\log(\det(\Sigma))| \right) $$
    
    \end{theorem}

    \begin{proof}
    In order to apply Theorem \ref{cdet}, we need to be sure that all the eigenvalues lie in $(\sigma_{min}, 1)$. In order to satisfy this condition, we can scale the matrix by a constant factor $c$, such that $\Sigma'=\Sigma/c$. In this way, $\log det(\Sigma') = \log \prod_i^d (\sigma_i/c)$. Therefore, $\log (det(\Sigma')) = \log (det (\Sigma)) - \log(c^d)$. This will allow us to recover the value of $\log det(\Sigma)$ by using Theorem \ref{cdet}. We apply the Theorem with precision $\epsilon=1/4$ to get an estimate $\gamma$ such that $\frac{1}{2} \leq \frac{\gamma} { \log(\det(\Sigma)) }  \leq \frac{3}{2}.$  Then, to have an estimate with absolute error $\epsilon$, we apply Theorem  \ref{cdet} with precision $\epsilon' =  \frac{\epsilon}{ 4 \gamma}$.  This gives us an estimate for $\log(\det(\Sigma))$ with error $2 \epsilon' \log(\det(\Sigma)) \leq \epsilon$ in time:
    $$ \widetilde{O} \left( \epsilon^{-2}\kappa(\Sigma)\log(1/\delta)nnz(\Sigma) {|\log(\det(\Sigma))|} \right).$$
    \end{proof}

    Now we can state the Lemma used to compute the responsibilities: a quantum algorithm for evaluating the exponent of a Gaussian distribution. 
    
        \begin{lemma}[Quantum Gaussian Evaluation]\label{lemma:gaussian}
            Suppose we have  stored in the QRAM a matrix $V \in \mathbb{R}^{n \times d}$, the centroid $\mu \in \mathbb{R}^d$ and the covariance matrix $\Sigma \in \mathbb{R}^{d \times d}$  of a multivariate Gaussian distribution $\phi(v | \mu, \Sigma)$, as well as an estimate for $\log(\det(\Sigma))$. Then for $\errgauss>0$, there exists a quantum algorithm that with probability $1-\gamma$ performs the mapping, 
            \begin{itemize}
                \item $U_{G, \errgauss } : \ket{i} \ket{0} \to \ket{i}\ket{\overline{s_i}}$  such that $|s_i - \overline{s_i}| < \epsilon_{1}$, where $s_i = - \frac{1}{2} ((v_i-\mu)^T\Sigma^{-1}(v_i-\mu) + d \log 2\pi + \log (det(\Sigma)))$  is the exponent for the Gaussian probability density function. 
            \end{itemize}
            The running time is $T_{G,\errgauss} = O\left(\frac{\kappa(\Sigma)\mu(\Sigma) \log (1/\gamma) }{\errgauss}\eta \right).$
            
        \end{lemma}
        
        \begin{proof}
            We use quantum linear algebra and inner product estimation to estimate the quadratic form $(v_i-\mu)^T\Sigma^{-1}(v_i-\mu)$ 
            to error $\epsilon_{1}$. First, we decompose the quadratic form as $v_i^T\Sigma^{-1}v_i - 2v_i^T\Sigma^{-1}\mu + \mu^T\Sigma^{-1}\mu $
            and separately approximate each term in the sum to error $\epsilon_{1}/4$, using Lemma \ref{lemma:quadratic forms}. and obtain an estimate for $\frac{1}{2} ((v_i-\mu)^T\Sigma^{-1}(v_i-\mu)$ within error $\epsilon_{1}$. 
            
            Recall that (through Lemma  \ref{lemma:det}) we also have an estimate of the log-determinant to error $\errgauss$. 				
            Thus we obtain an approximation 
            for $- \frac{1}{2} ((v_i-\mu)^T\Sigma^{-1}(v_i-\mu) + d \log 2\pi + \log (\det(\Sigma)))$ within error $2\epsilon_{1}$. 
            The running time for this computation is
            $ O\left(\frac{\kappa(\Sigma)\mu(\Sigma) \log (1/\gamma) \log (1/\errmult)}{\errgauss}\eta \right)$. %

        \end{proof}

        Using controlled operations it is simple to extend the previous Theorem to work with multiple Gaussians distributions $(\mu_j,\Sigma_j)$. That is, we can control on a register $\ket{j}$ to do $\ket{j}\ket{i} \ket{0} \mapsto \ket{j}\ket{i}\ket{\phi(v_i|\mu_j,\Sigma_j)}$. 
        In the next Lemma we will see how to obtain the responsibilities $r_{ij}$ using the previous Theorem and standard quantum circuits for doing arithmetic, controlled rotations, and amplitude amplification.
        
        \begin{lemma}[Calculating responsibilities]\label{lemma:responsibility}
            Suppose we have quantum access to a GMM with parameters $\gamma^{t}=(\theta^{t}, \bm \mu^{t}, \bm \Sigma^{t})$.
    There are quantum algorithms that can:
            \begin{enumerate}
                \item  Perform the mapping $ \ket{i}\ket{j}\ket{0} \mapsto \ket{i}\ket{j}\ket{\overline{r_{ij}}}$   
                such that $|\overline{r_{ij}} - r_{ij}| \leq \errgauss$ with high probability in time:
                $$T_{R_1,\errgauss} =  \widetilde{O}(k^{1.5} \times T_{G,\errgauss})  $$
                \item For a given $j \in [k]$, construct state $\ket{\overline{R_j}}$ such that  $\norm{\ket{\overline{R_j}} - \frac{1}{\sqrt{Z_j}}\sum\limits_{i=0}^n r_{ij}\ket{i}} < \errgauss$ where $Z_j=\sum\limits_{i=0}^n r_{ij}^2$ with high probability in time:
                $$T_{R_2,\errgauss} = \widetilde{O}(k^{2} \times T_{R_1,\errgauss}) $$ 
            \end{enumerate}
        \end{lemma}
    
        \begin{proof}
            For the first statement, let's recall the definition of responsibility: $ r_{ij} = \frac{\theta_j \phi(v_i; \mu_j, \Sigma_j )}{\sum_{l=1}^k \theta_l \phi(v_i; \mu_l, \Sigma_l)}$. With the aid of $U_{G, \errgauss}$ of Lemma \ref{lemma:gaussian} we can estimate $\log (\phi(v_i|\mu_j, \Sigma_j)) $ for all $j$ up to additive error $\errgauss$, and then using the current estimate of $\theta^t$,  we can calculate the responsibilities and create the state, 
            $$ \frac{1}{\sqrt{n}} \sum_{i=0}^{n} \ket{i} \Big( \bigotimes_{j=1}^k \ket{j}\ket{ \overline{ \log (\phi(v_i|\mu_j, \Sigma_j) }} \Big) \otimes  \ket{\overline{r_{ij}} } . $$
            The estimate $\overline{r_{ij}}$ is computed by evaluating a weighted softmax function with arguments $\overline{ \log (\phi(v_i|\mu_j, \Sigma_j) }$ for $j\in [k]$. The estimates $\overline{ \log (\phi(v_i|\mu_j, \Sigma_j) }$ are then uncomputed. 
            The runtime of the procedure is given by calling $k$ times Lemma $\ref{lemma:gaussian}$ for Gaussian estimation (the arithmetic operations to calculate the responsibilities are absorbed). 
            
            Let us analyze the error in the estimation of $r_{ij}$. The responsibility $r_{ij}$ is a softmax function with arguments $\log (\phi(v_i|\mu_j, \Sigma_j)) $ that are computed upto error $\errgauss$ using Lemma \ref{lemma:gaussian}. As the softmax function has a Lipschitz constant $K\leq \sqrt{2}$ by Lemma \ref{respsoftmaxed}, we choose precision for Lemma \ref{lemma:gaussian} to be $\errgauss/\sqrt{2k}$ to get 
            the guarantee $|\overline{r_{ij}} - r_{ij} | \leq \errgauss$. Thus, the total cost of this step is $T_{R_1,\errgauss} = k^{1.5} T_{G,\errgauss} $.

            Now let's see how to encode this information in the amplitudes, as stated in the second claim of the Lemma. We estimate the responsibilities $r_{ij}$ to some precision $\epsilon$ and perform a controlled rotation on an ancillary qubit to obtain, 
            
            \begin{equation}\label{postselectme}
                \frac{1}{\sqrt{n}} \ket{j} \sum_{i=0}^{n} \ket{i} \ket{\overline{r_{ij}}}\Big(\overline{r_{ij}}\ket{0} + \sqrt{1-\overline{r_{ij}}^2  }\ket{1}\Big).
            \end{equation}
            We then undo the circuit on the second register and perform amplitude amplification on the rightmost auxiliary qubit being $\ket{0}$ to get
                     $\ket{\overline{R_j}}:=\frac{1}{\norm{\overline{R_j} }}\sum_{i=0}^n \overline{r_{ij}}\ket{i}$. The runtime for amplitude amplification on this task is $O(T_{R_1,\epsilon} \cdot \frac{\sqrt{n}}{ \;\norm{\overline{R_j} } \; })$. 
            
            Let us analyze the precision $\epsilon$ required to prepare $\ket{\overline{R_j}}$ such that $\norm{\ket{R_j} - \ket{\overline{R_j}}} \leq \epsilon_{1}$. As we have estimates $|r_{ij}-\overline{r_{ij}}|<\epsilon$ for all $i, j$, the 
            $\ell_{2}$-norm error $\norm{ R_j - \overline{R_{j}}} = \sqrt{\sum_{i=0}^n | r_{ij} -  \overline{r_{ij}} |^2 }< \sqrt{n}\epsilon $. 
            Applying Claim \ref{kereclaim}, the error for the normalized vector $\ket{R_j}$ can be bounded as $ \norm{\ket{R_j} - \ket{\overline{R_j}}} < \frac{\sqrt{2n}\epsilon}{\norm{R_j}}  $. 
            By the Cauchy-Schwarz inequality we have that $\norm{R_j} \geq  \frac{\sum_i^n r_{ij}}{\sqrt{n}}$. We can use this to obtain a bound $\frac{\sqrt{n}}{\norm{R_j}}<\frac{\sqrt{n}}{\sum_i r_{ij}}\sqrt{n} = O(1/k)$, 
            using the dataset assumptions in the main manuscript. If we choose $\epsilon$ such that $\frac{\sqrt{2n}\epsilon}{\norm{R_j}} < \errgauss$, that is $\epsilon \leq \errgauss/k$ then our runtime becomes $T_{R_2,\errgauss} := \widetilde{O}(k^{2} \times T_{R_1,\errgauss})$.

        \end{proof}

    \subsubsection{Maximization}
        Now we need to get an updated estimate for the parameters of our model. At each iteration of QEM we build a quantum state proportional to the updated parameters of the model, and then recover them. Once the new model has been obtained, we update the QRAM such that we get quantum access to the model $\gamma^{t+1}$. The possibility  to estimate $\theta$ comes from a call to the unitary we built to compute the responsibilities, and amplitude amplification.
        
        \begin{lemma}[Computing $\theta^{t+1}$]\label{lemma:theta}
            We assume quantum access to a GMM with parameters $\gamma^{t}$ and let $\delta_\theta > 0$ be a precision parameter. There exists an algorithm that estimates $ \overline{\theta}^{t+1} \in \mathbb{R}^k$ such that $\norm{\overline{\theta}^{t+1}-\theta^{t+1}}\leq \delta_\theta$ in time 
            $T_\theta = O\left(k^{3.5} \eta^{1.5} \frac{ \kappa(\Sigma)\mu(\Sigma) }{\delta_\theta^2}  \right)$. 
        \end{lemma}
        	\begin{proof}
		An estimate of $\theta^{t+1}_j$ can be recovered from the following operations. First, we use Lemma \ref{lemma:responsibility} (part 1) to compute the responsibilities to error $\errgauss$, and then perform the following mapping, which consists of a controlled rotation on an auxiliary qubit:
    \begin{align}    \frac{1}{\sqrt{nk}}\sum_{\substack{i =1 \\j =1}}^{n,k}\ket{i}\ket{j}\ket{\overline{r_{ij}}^{t} } \mapsto \nonumber \\  \frac{1}{\sqrt{nk}}\sum_{\substack{i =1 \\j =1}}^{n,k}\ket{i}\ket{j}\Big( \sqrt{\overline{r_{ij}}^{t}}\ket{0} + \sqrt{1-\overline{r_{ij}}^{t}}\ket{1} \Big)  \end{align}
		The previous operation has a cost of $T_{R_1,\errgauss}$, and the probability of getting $\ket{0}$ is
		$p(0) = \frac{1}{nk} \sum_{i=1}^{n}\sum_{j=1}^k r_{ij}^{t} = \frac{1}{k}$.
		
		Recall that $\theta_j^{t+1} = \frac{1}{n}\sum_{i=1}^n r^{t}_{ij}$ by definition. 
		Let $Z_j = \sum_{i=1}^n \overline{r_{ij}}^{t}$ and define state $\ket{\sqrt{R_j}} = \left(\frac{1}{\sqrt{Z_j}} \sum_{i=1}^n \sqrt{ \overline{r_{ij}}^{t} }\ket{i}\right)\ket{j}  $. 
		After amplitude amplification on $\ket{0}$ we have the state,
		\begin{align}\label{eq:theta} 
			\ket{\sqrt{R}} &:= \frac{1}{\sqrt{n}} \sum_{\substack{i =1 \\j =1}}^{n,k} \sqrt{  \overline{r_{ij}}^{t}  }\ket{i}\ket{j}\nonumber \\
			&= \sum_{j=1}^k\sqrt{\frac{Z_j}{n}} \left(\frac{1}{\sqrt{Z_j}} \sum_{i=1}^n \sqrt{ \overline{r_{ij}}^{t} }\ket{i} \right) \ket{j}   \nonumber\\
			&=  \sum_{j=1}^k\sqrt{\overline{\theta_j}^{t+1} }\ket{\sqrt{R_j}}\ket{j}. %
		\end{align}
		The probability of obtaining outcome $\ket{j}$ if the second register is measured in the standard basis is $p(j)=\overline{\theta_j}^{t+1}$. 
		
		An estimate for $\theta_j^{t+1}$ with precision $\epsilon$ can be obtained by either sampling  the last register, or by performing amplitude estimation to estimate each of the values $\theta^{t+1}_j$ for $j \in [k]$. Sampling requires $O(\epsilon^{-2})$ samples by the Chernoff bounds, but does not incur any dependence on $k$. In this case, as the number of cluster $k$ is relatively small compared to $1/\epsilon$, we chose to do amplitude estimation to estimate all $\theta^{t+1}_j$ for $j \in [k]$ to error $\epsilon/\sqrt{k}$ in time, 
		
	\begin{equation} T_\theta := O\left(k\cdot \frac{\sqrt{k}T_{R_1,\errgauss}}{\epsilon} \right).
	\end{equation}
		
		Let's analyze the error in the estimation of $\theta_j^{t+1}$. For the error due to responsibility estimation by Lemma \ref{lemma:responsibility} we have  $|\overline{\theta_j}^{t+1} - \theta_j^{t+1}| 
		= \frac{1}{n} \sum_i | \overline{r_{ij}}^{t} - r_{ij}^{t}| \leq \epsilon_{1}$ for all $j \in [k]$, implying that  $\norm{ \overline{\theta}^{t+1} - \theta^{t+1} }  \leq \sqrt{k}  \epsilon_{1}$. 
		The total error in $\ell_{2}$ norm due to Amplitude estimation is at most $\epsilon$ as it estimates each coordinate of $\overline{\theta_j}^{t+1}$ to error $\epsilon/\sqrt{k}$.

		Using the triangle inequality, we have the total error is at most $\epsilon  + \sqrt{k}\errgauss$. 
		As we require that the final error be upper bounded as $\norm{\overline{\theta}^{t+1} - \theta^{t+1}} < \delta_\theta$, we choose parameters 
		$\sqrt{k}\errgauss < \delta_\theta/2  \Rightarrow \errgauss < \frac{\delta_\theta}{2\sqrt{k}}$ and $\epsilon < \delta_\theta/2$. With these parameters, the overall running 
		time of the quantum procedure is
		$ T_\theta = O(k^{1.5} \frac{T_{R_1,\errgauss}}{\epsilon}) = O\left(k^{3.5} \frac{ \eta^{1.5}\cdot \kappa(\Sigma)\mu(\Sigma) }{\delta_\theta^2}  \right)$.

			\end{proof}
        
        We use quantum linear algebra to transform the uniform superposition of responsibilities of the $j$-th mixture into the new centroid of the $j$-th Gaussian. Let $R_j^{t} \in \R^{n}$ be the vector of responsibilities for a Gaussian $j$ at iteration $t$. The following claim relates the vectors  $R_j^{t}$ to the centroids $\mu_j^{t+1}$ and its proof is straightforward. 
        
        \begin{claim}\label{claimz}
            Let $R_j^t \in \mathbb{R}^n$ be the vector of responsibilities of the points for the Gaussian $j$ at time $t$, i.e. $(R_j^t)_{i} = r_{ij}^{t}$. Then $\mu_j^{t+1} \leftarrow \frac{\sum_{i=1}^n r^{t}_{ij} v_i }{ \sum_{i=1}^n r^{t}_{ij}} = \frac{V^T R^t_j}{n\theta_j}$.
        \end{claim}

        From this Claim, we derive the procedure to estimate the new centroids $\mu_j^{t+1}$: we use Lemma \ref{lemma:responsibility} along with quantum access to the matrix $V$. 
    
        \begin{lemma}[Computing $\bm \mu_j^{t+1}$]\label{lemma:centroids}
            We assume we have quantum access to a GMM with parameters $\gamma^{t}$. For a precision parameter $\delta_\mu > 0$, there is a quantum algorithm that calculates $\{ \overline{\mu_j}^{t+1} \}_{j=1}^k$ such that for all $j\in [k]$ $\norm{\overline{\mu_j}^{t+1} -\mu_j^{t+1}}\leq \delta_\mu$  in time
    $ T_\mu = \widetilde{O}\left(    \frac{k d\eta \kappa(V) (\mu(V) + k^{3.5}\eta^{1.5}\kappa(\Sigma)\mu(\Sigma))}{\delta_{\mu}^3}  \right)$				
    
        \end{lemma}
        
        \begin{proof}
            The new centroid $\mu_j^{t+1}$ is estimated by first creating an approximation of the state $\ket{R_j^t}$ up to error $\errgauss$ in the $\ell_{2}$-norm using part 2 of Lemma \ref{lemma:responsibility}. We then use the quantum linear algebra algorithms in Theorem \ref{qla} to multiply $R_j$ by $V^T$, and compute a state $\ket{\overline{\mu_j}^{t+1}}$ along with an estimate for the norm $\norm{V^TR_j^t} = \norm{\overline{\mu_j}^{t+1}}$ with error $\errnorms$. The last step of the algorithm consists in estimating the unit vector $\ket{\overline{\mu_j}^{t+1}}$ with precision $\errtom$ using tomography. Considering that the tomography depends on $d$, which we expect to be bigger than the precision required by the norm estimation, we can assume that the runtime of the norm estimation is absorbed. Thus, we obtain: $ \widetilde{O}\left( k  \frac{d}{\errtom^2} \cdot  \kappa(V) \left( \mu(V) + T_{R_2,\epsilon_1} \right)  \right) $.
                    
            Let's now analyze the total error in the estimation of the new centroids, which we want to be $\delta_{\mu}$. For this purpose, we use Claim \ref{quattrocinque}, and choose parameters such that $2\sqrt{\eta}(\epsilon_{tom}+\epsilon_{norm})=\delta_\mu$. Since the error $\errmult$ for quantum linear algebra appears as a logarithmic factor in the running time, we can choose $\errmult \ll\errtom$ without affecting the runtime. 
            
    Let $\overline{\mu}$ be the classical unit vector obtained after quantum tomography, and $\widehat{\ket{\mu}}$ be the state produced by the quantum linear algebra procedure starting with an approximation of $\ket{R_j^t}$. 
    Using the triangle inequality we have $ \norm{\ket{\mu} - \overline{\mu} } < \norm{ \overline{\mu} - \widehat{\ket{\mu}}} + \norm{\widehat{\ket{\mu}} - \ket{\mu}} < \errtom + \errgauss < \delta_{\mu}/2\sqrt{\eta} $. 
    The errors for the norm estimation procedure can be bounded similarly as $ | \norm{\mu} - \overline{\norm{\mu}} | < | \norm{\mu} - \widehat{\norm{\mu}} | + | \overline{\norm{\mu}} - \widehat{\norm{\mu}} |  <  \errnorms + \errgauss \leq \delta_{\mu}/2\sqrt{\eta}$. We therefore choose parameters  $\errtom = \errgauss  = \errnorms \leq \delta_{\mu}/4\sqrt{\eta}$. Since the amplitude estimation step we use for estimating the norms does not depends on $d$, which is expected to dominate the other parameters, we omit the amplitude estimation step. Substituting for $T_{R_{2}, \delta_{\mu}}$, we have the more concise expression for the running time of:
    \begin{equation}
        \widetilde{O}\left(    \frac{k d\eta \kappa(V) (\mu(V) + k^{3.5}\eta^{1.5}\kappa(\Sigma)\mu(\Sigma))}{\delta_{\mu}^3}  \right)
    \end{equation}

        \end{proof}

        From the ability to calculate responsibility and indexing the centroids, we derive the ability to reconstruct the covariance matrix of the Gaussians as well. Again, we use quantum linear algebra subroutines and tomography to recover an approximation of each $\Sigma_j$. Recall that we have defined the matrix $V' \in \mathbb{R}^{n \times d^2}$ where the $i$-th row of $V'$ is defined as $\text{vec}[v_iv_i^T]$. Note that the quantum states corresponding to the rows of $V'$ can be prepared as $\ket{i} \ket{0} \ket{0} \to \ket{i} \ket{v_{i}} \ket{v_{i}}$, using twice the procedure for creating the rows of $V$.  
    
    \begin{lemma}[Computing $ \Sigma_j^{t+1}$]\label{lemma:sigma}
        Given quantum access to a GMM with parameters $\gamma^{t}$. We also have computed estimates $\overline{\mu_j}^{t+1}$ of all centroids such that $\norm{ \overline{\mu_j}^{t+1} - \mu_j^{t+1}  }\leq \delta_\mu$ for precision parameter $\delta_\mu > 0$. Then, there exists a quantum algorithm that outputs estimates for the new covariance matrices $\{ \overline{\Sigma}^{t+1}_j \}_{j=1}^k$ such that $\norm{\Sigma_j^{t+1}- \overline{\Sigma}^{t+1}_j}_F \leq  \delta_\mu\sqrt{\eta}$ with high probability, in time, 
        $$
        T_\Sigma := \widetilde{O} \Big( \frac{kd^2 \eta\kappa(V)(\mu(V')+\eta^2k^{3.5}\kappa(\Sigma)\mu(\Sigma))}{\delta_{\mu}^3} \Big)	
        $$
    
    \end{lemma}
    
    \begin{proof}
        It is simple to check, that the update rule of the covariance matrix during the maximization step can be reduced to \cite{murphy2012machine}{[Exercise 11.2]}:
            
            \begin{align}\label{covmat}\Sigma_j^{t+1} \leftarrow \frac{\sum_{i=1}^n r_{ij}  (v_i - \mu_j^{t+1})(v_i - \mu_j^{t+1})^T}{\sum_{i=1}^n r_{ij}} & \\
            = \frac{\sum_{i=1}^n r_{ij}v_iv_i^T }{n\theta_j} - \mu_j^{t+1}(\mu_j^{t+1})^T & \\
            =\Sigma_j'- \mu_j^{t+1}(\mu_j^{t+1})^T &
            \end{align}
            
    First, note that we can use the previously obtained estimates of the centroids to compute the outer product $\mu_j^{t+1} (\mu_j^{t+1})^T$ with error $\delta_\mu \norm{ \mu} \leq \delta_\mu \sqrt{\eta}$. The error in the estimates of the centroids is $\overline{\mu} = \mu + e$ where $e$ is a vector of norm $\delta_\mu$. Therefore $\norm{\mu \mu^T - \overline{\mu}\:\overline{\mu}^T} < 2\sqrt{\eta}\delta_\mu + \delta_\mu^2 \leq 3\sqrt{\eta}\delta_\mu$. Because of this, we allow an error of $\sqrt{\eta}\delta_\mu$ also for the term $\Sigma'_j$. 
    Now we discuss the procedure for estimating $\Sigma_j'$, which we split into an estimate of $\ket{\text{vec}[\Sigma_j']}$ and its norm $\norm{\text{vec}[\Sigma_j']}$. We start by using quantum access to the norms and part 1 of Lemma \ref{responsibility} with error $\epsilon_1$. For a cluster $j$, we create the state:
    
    $$ \ket{j}\frac{1}{\sqrt{n}}\sum_i^n \ket{i}\ket{r_{ij}}\ket{\norm{v_i}}\left(\frac{r_{ij}\norm{v_i}}{\sqrt{\eta}} \ket{0} + \gamma\ket{1} \right) $$ 
    
    We proceed by using amplitude amplification on the right-most qubit being $\ket{0}$, which takes time $O(R_{R_1,\epsilon_1}\frac{\sqrt{n\eta}}{\norm{V_R}})$, where $\norm{V_R}$ is $\sqrt{\sum_i r_{ij}^2\norm{v_i}^2}$. Successively, we use quantum access on the vectors $v_i$ and we obtain the following state:
    $$
    \frac{1}{V_R}\sum_i r_{ij}\norm{v_i}\ket{i}\ket{v_i}
    $$
    Now we can apply quantum linear algebra subroutine, multiplying the first register with the matrix $V^T$. This will lead us to the desired state $\ket{\Sigma_j'}$, along with an estimate of its norm with relative precision $\epsilon_2$
    
    As the runtime for the norm estimation $\frac{\kappa(V)(\mu(V) + T_{R_2,\epsilon_1}))\log(1/\epsilon_{mult})}{\epsilon_{norms}}$ does not depend on $d$, we consider it smaller than the runtime for performing tomography. Thus, the runtime for this operation is:
            $$O(\frac{d^2\log d}{\epsilon^2_{tom}}\kappa(V)(\mu(V) + T_{R_2, \epsilon_1}))\log(1/\epsilon_{mult})). $$
        
    Let's analyze the error of this procedure. We want a matrix $\overline{\Sigma_j'}$ that is $\sqrt{\eta}\delta_\mu$-close to the correct one: $\norm{\overline{\Sigma_j'} - \Sigma'_j}_F = \norm{\text{vec}\overline{[\Sigma_j']} - \text{vec}[\Sigma'_j]}_2 < \sqrt{\eta}\delta_\mu$. Again, the error due to matrix multiplication can be taken as small as necessary, since is inside a logarithm. From Claim \ref{quattrocinque}, we just need to fix the error of tomography and norm estimation such that 		$\eta(\epsilon_{unit} + \epsilon_{norms}) < \sqrt{\eta}\delta_{\mu} $ where we have used $\eta$ as an upper bound on $\norm{\Sigma_j}_F$. 
    For the unit vectors, we require
            $ \norm{\ket{\Sigma'_j} - \overline{\ket{\Sigma'_j}} } \leq \norm{ \overline{\ket{\Sigma'_j}} - \widehat{\ket{\Sigma'_j}}} + \norm{\widehat{\ket{\Sigma'_j}} - \ket{\Sigma'_j}} < \errtom + \epsilon_1 \leq \eta\epsilon_{unit} \leq \frac{\delta_{\mu}\sqrt{\eta}}{2}$, where 
            $\overline{\ket{\Sigma'_j}} $ is the error due to tomography and $\widehat{\ket{\Sigma'_j}}$ is the error due to the responsibilities in Lemma \ref{lemma:responsibility}. For this inequality to be true, we choose $\errtom = \epsilon_1 < \frac{\delta_\mu/\sqrt{\eta}}{4}$. 
            
            The same argument applies to estimating the norm $\norm{\Sigma_j'}$ with relative error :
            $ |\norm{\Sigma'_j} - \overline{\norm{\Sigma'_j}} |  \leq  | \overline{\norm{\Sigma'_j}} - \widehat{\norm{\Sigma'_j}}| + |\widehat{\norm{\Sigma'_j}} - \norm{\Sigma'_j}| < \epsilon_2 + \epsilon_1 \leq \delta_{\mu}/2\sqrt{\eta}$ (where here $\epsilon_2$ is the error of the amplitude estimation step used in Theorem \ref{qla} and $\epsilon_1$ is the error in calling Lemma \ref{lemma:responsibility}. Again, we choose $\epsilon_2=\epsilon_1 \leq \frac{\delta_\mu/\sqrt{\eta}}{4}$. 
            
                Since the tomography is more costly than the amplitude estimation step, we can disregard the runtime for the norm estimation step. As this operation is repeated $k$ times for the $k$ different covariance matrices, the total runtime of the whole algorithm is given by $\widetilde{O} \Big( \frac{kd^2 \eta\kappa(V)(\mu(V)+\eta^2k^{3.5}\kappa(\Sigma)\mu(\Sigma))}{\delta_{\mu}^3} \Big)		$.		
    Let us also recall that for each of new computed covariance matrices, we use Lemma \ref{lemma:det} to compute an estimate for their log-determinant and this time can be absorbed in the time $T_\Sigma$.
        \end{proof}

        \subsubsection{Quantum estimation of log-likelihood}
        Now we are going to state how to get an estimate of the log-likelihood using a quantum procedure and access to a GMM model. Because of the error analysis, in the quantum algorithm is more conveniente to estimate $\mathbb{E}[ p(v_i ; \gamma^t)] = \frac{1}{n}\sum_{i=1}^n  p(v_i; \gamma) $.  From this we can estimate an upper bound on the $\log$-likelihood as $n \log \mathbb{E}[p(v_i)] = \sum_{i=1}^n \log \mathbb{E}[p(v_i)] \geq \sum_{i=1}^n \log p(v_i) = \ell(\gamma; V)$.

        \begin{lemma}[Quantum estimation of likelihood]\label{lemma:likelihoodestimation}
             We assume we have quantum access to a GMM with parameters $\gamma^{t}$. For $\errlikelihood > 0$, there exists a quantum algorithm that		
                estimates  $\mathbb{E}[p(v_i  ; \gamma^{t} )]$ with absolute error $\errlikelihood$ 
             in time
            $$ T_\ell = \widetilde{O}\left( k^{1.5}\eta^{1.5}  \frac{\kappa(\Sigma)\mu(\Sigma)}{\errlikelihood^2} \right)$$		
        \end{lemma}

        \begin{proof}
            We obtain the likelihood from the ability to compute the value of a Gaussian distribution and quantum arithmetic. Using the mapping of Lemma $\ref{lemma:gaussian}$ with precision $\errgauss$, we can compute $\phi(v_i|\mu_j, \Sigma_j)$ for all the Gaussians, that is $ \ket{i} \bigotimes_{j=0}^{k-1} \ket{j}\ket{\overline{p(v_i|j;\gamma_j)}}$. Then, by knowing $\theta$, and by using quantum arithmetic we can compute in a register the mixture of Gaussian's: $p(v_i; \gamma)  = \sum_{j\in[k]} \theta_j p(v_i|j;\gamma )$. We now drop the notation for the model $\gamma$ and write $p(v_i)$ instead of $p(v_i; \gamma)$. Doing the previous calculations quantumly leads to the creation of the state $\ket{i}\ket{p(v_i)}$. 
            We perform the mapping $\ket{i}\ket{p(v_i)}\mapsto \ket{i}\left( \sqrt{p(v_i)|}\ket{0} + \sqrt{1 - p(v_i)}\ket{1} \right)$ and estimate $p(\ket{0}) \simeq \mathbb{E}[p(v_i)]$ with amplitude estimation on the ancilla qubit. 
            To get a $\errlikelihood$-estimate of $p(0)$ we need to decide the precision parameter we use for estimating $\overline{p(v_i|j;\gamma)}$ and the precision required by amplitude estimation. 
            Let $\overline{p(0)}$ be the $\errgauss$-error introduced by using Lemma \ref{lemma:gaussian} and $\widehat{p(0)}$ the error introduced by amplitude estimation. Using triangle inequality we set $\norm{p(0) - \widehat{p(0)}} < \norm{\widehat{p(0)} - \overline{p(0)}} + \norm{\overline{p(0)} - p(0)} < \errlikelihood $.  
            
            To have $ | p(0) - \widehat{p(0)}| < \errlikelihood$, we should set $\errgauss$ such that $ | \overline{p(0)} - p(0) | < \errlikelihood/4$, and we set the error in amplitude estimation and in the estimation of the probabilities to be $\errlikelihood/2$. The runtime of this procedure is therefore:
            $$ \widetilde{O}\left(k \cdot T_{G,\errlikelihood} \cdot \frac{1}{\errlikelihood \sqrt{p(0)}}\right) = \widetilde{O}\left(k^{1.5}\eta^{1.5} \cdot \frac{\kappa(\Sigma)\mu(\Sigma)}{\errlikelihood^2 }\right)$$

          \end{proof}

        \subsubsection{QEM for GMM}
        Putting together all the previous Lemmas, we write the main result of the work.
    
    \begin{theorem}[QEM for GMM]\label{th:qgmm}
        We assume we have quantum access to a GMM with parameters $\gamma^t$. For parameters $\delta_\theta, \delta_\mu, \epsilon_\tau > 0$, the running time of one iteration of the Quantum Expectation-Maximization (QEM) algorithm is 
    $$O( T_\theta + T_\mu + T_\Sigma + T_\ell),$$
    for $	T_\theta   =  \widetilde{O}\left(k^{3.5} \eta^{1.5} \frac{ \kappa(\Sigma)\mu(\Sigma) }{\delta_\theta^2}  \right)$,\\  $T_\mu  =  \widetilde{O}\left(    \frac{k d\eta \kappa(V) (\mu(V) + k^{3.5}\eta^{1.5}\kappa(\Sigma)\mu(\Sigma))}{\delta_{\mu}^3}  \right)$,\\ $T_\Sigma  =  \widetilde{O} \Big( \frac{kd^2 \eta\kappa^2(V)(\mu(V')+\eta^2k^{3.5}\kappa(\Sigma)\mu(\Sigma))}{\delta_{\mu}^3} \Big)$ and\\ $T_\ell  =  \widetilde{O}\left( k^{1.5}\eta^{1.5}  \frac{\kappa(\Sigma)\mu(\Sigma)}{\errlikelihood^2} \right).$
    
    For parameter that are expected to be predominant in the runtime, $d$ and $\kappa(V)$, the dominant term is $T_\Sigma$.
    \end{theorem}
    
    The proof follows directly from the previous lemmas. Note that the cost of the whole algorithm is given by repeating the expectation and the maximization steps several times, until the threshold on the log-likelihood is reached. The expression of the runtime can be simplified by observing that the cost of performing tomography on the covariance matrices $\Sigma_j$ dominates the runtime.

\section{Experimental Results}\label{expevidence}
    In this section, we present the results of some experiments on real datasets to estimate the runtime of the algorithm, and bound the value of the parameters that governs the runtime, like $\kappa(\Sigma)$, $\kappa(V)$, $\mu(\Sigma)$, $\mu(V)$, $\delta_\theta$, and $\delta_\mu$, and we give heuristic for dealing with the condition number. We can put a threshold on the condition number of the matrices $\Sigma_j$, by discarding singular values which are smaller than a certain threshold. This might decrease the runtime of the algorithm without impacting its performances. This is indeed done often in classical machine learning models, since discarding the eigenvalues smaller than a certain threshold might even improve upon the metric under consideration (i.e. often the accuracy), by acting as a form of regularization \citep[Section 6.5]{murphy2012machine}. This is equivalent to limiting the eccentricity of the Gaussians. We can do similar considerations for putting a threshold on the condition number of the dataset $\kappa(V)$. Recall that the value of the condition number of the matrix $V$ is approximately
    $1/\min( \{\theta_1,\cdots, \theta_k\}\cup \{ d_{st}(\mathcal{N}(\mu_i, \Sigma_i), \mathcal{N}(\mu_j, \Sigma_j)) | i\neq j \in [k] \})$, where $d_{st}$ is the statistical distance between two Gaussian distributions \cite{kalai2012disentangling}. We have some choice in picking the definition for $\mu$: in previous experiments it has been found that choosing the maximum $\ell_1$ norm of the rows of $V$ lead to values of $\mu(V)$ around $10$ for the MNIST dataset \cite{KL18, kerenidis2019q}. Because of the way $\mu$ is defined, its value will not increase significantly as we add vectors to the training set. In case the matrix $V$ can be clustered with high-enough accuracy by distance-based algorithms like k-means, it has been showed that the Frobenius norm of the matrix is proportional to $\sqrt{k}$, that is, the rank of the matrix depends on the number of different classes contained in the data. Given that EM is just a more powerful extension of k-means, we can rely on similar observations too. Usually, the number of features $d$ is much more than the number of components in the mixture, i.e. $d \gg k$, so we expect $d^2$ to dominate the $k^{3.5}$ term in the cost needed to estimate the mixing weights, thus making $T_\Sigma$ the leading term in the runtime. We expect this cost to be be mitigated by using $\ell_\infty$ form of tomography but we defer further experiment for future research.
    
As we said, the quantum running time saves the factor that depends on the number of samples and introduces a number of other parameters. Using our experimental results we can see that when the number of samples is large enough one can expect the quantum running time to be faster than the classical one. One may also save some more factors from the quantum running time with a more careful analysis.

   To estimate the runtime of the algorithm, we need to gauge the value of the parameters $\delta_\mu$ and $\delta_\theta$, such that they are small enough so that the likelihood is perturbed less than $\epsilon_\tau$, but big enough to have a fast algorithm. We have reasons to believe that on well-clusterable data, the value of these parameters will be large enough, such as not to impact dramatically the runtime. A quantum version of k-means algorithm has already been simulated on the MNIST dataset under similar assumptions \citep{kerenidis2019q}. The experiment concluded that, for datasets that are expected to be clustered nicely by this kind of clustering algorithms, the value of the parameters $\delta_\mu$ did not decrease by increasing the number of samples nor the number of features. There, the value of $\delta_\mu$ (which in their case was called just $\delta$) has been kept between $0.2$ and $0.5$, while retaining a classification accuracy comparable to the classical k-means algorithm. We expect similar behaviour in the GMM case, namely that for large datasets the impact on the runtime of the errors $(\delta_\mu, \delta_\theta)$ does not cancel out the exponential gain in the dependence on the number of samples, and we discuss more about this in the next paragraph. The value of $\epsilon_\tau$ is usually (for instance in scikit-learn \citep{scikit-learn} ) chosen to be $10^{-3}$. We will see that the value of $\eta$ has always been 10 on average, with a maximum of 105 in the experiments.  

    \begin{table}[h]
        \centering
     
        \begin{tabular}{c|c|c|c|c|}
        \cline{2-5}
 & \multicolumn{2}{c|}{MAP}  & \multicolumn{2}{c|}{ML}   \\ \cline{2-5} 
     & avg          & max        & avg         & max         \\ \hline
        \multicolumn{1}{|c|}{$\norm{\Sigma}_2$}                & 0.22         & 2.45       & 1.31        & 3.44        \\ \hline
        \multicolumn{1}{|c|}{$|\log \det(\Sigma)|$}                                 & 58.56        & 70.08      & 14.56       & 92.3        \\ \hline
        \multicolumn{1}{|c|}{$\kappa^*(\Sigma)$} & 4.21         & 50         & 15.57       & 50          \\ \hline
        \multicolumn{1}{|c|}{$\mu(\Sigma)$}                       & 3.82         & 4.35       & 2.54        & 3.67        \\ \hline
        \multicolumn{1}{|c|}{$\mu(V)$}                                           & 2.14         & 2.79       & 2.14        & 2.79        \\ \hline
        \multicolumn{1}{|c|}{$\kappa(V)$}                                        & 23.82        & 40.38      & 23.82       & 40.38       \\ \hline
        \end{tabular}
            \caption{We estimate some of the parameters of the VoxForge \citep{voxforge} dataset.  Each model is the result of the best of 3 different initializations of the EM algorithm. The first and the second rows are the maximum singular values of all the covariance matrices, and the absolute value of the log-determinant. The column $\kappa^*(\Sigma)$ shows the condition number after the smallest singular values have been discarded.}        \label{tab:results}

        \end{table}

        \section{Experiments} 
        We  analyzed a dataset which can be fitted well with the EM algorithm \citep{reynolds2000speaker,appliedgmm,voxforge}. Specifically, we used EM to do speaker recognition: the task of recognizing a speaker from a voice sample, having access to a training set of recorded voices of all the possible speakers. The training set consist in 5 speech utterances for 38 speakers (i.e. clips of a few seconds of voice speech). For each speaker, we extract the mel-frequency cepstral coefficients (MFCC) of the utterances \cite{reynolds2000speaker}, resulting in circa 5000 vectors of 40 dimensions. This represent the training set for each speaker. 
        A speaker is then modeled with a mixture of 16 different diagonal Gaussians. The test set consists of other 5 or more unseen utterances for each of the same speakers. To label an utterance with a speaker, we compute the log-likelihood of the utterance for each trained model. The label consist in the speaker with highest log-likelihood. The experiment has been carried in form of classical simulation on a laptop computer. We repeated the experiment using a perturbed model, where we added some noise to the GMM at each iteration of the training, as in Definition \ref{def:approxgmm}. Then we measured the accuracy of the speaker recognition task. At last, we measured condition number, the absolute value of the log-determinant, and the value of $\mu(V)$ and $\mu(\Sigma)$. In this way we can test the stability and accuracy of the approximate GMM model introduced in Section \ref{gmmsection}, under the effect of noise. For values of $\delta_\theta = 0.038$, $\delta_\mu=0.5$, we correctly classified 98.7\% utterances. The baseline for ML estimate of the GMM is of $97.1\%$. We attribute the improved accuracy to the regularizing effect of the threshold and the noise, as the standard ML estimate is likely to overfit the data. We report the results of the measurement in Table \ref{tab:results}.

        We used a subset of the voices that can be found on VoxForge \citep{voxforge}. The training set consist in at 5 speech utterances from 38 speakers. An utterance is a wav audio clips of a few seconds of voice speech.  In order to proceed with speech recognition from raw wav audio files, we need to proceed with classical feature extraction procedures. In the speech recognition community is common to extract from audio the Mel Frequency Cepstrum Coefficients (MFCCs) features \citep{reynolds2000speaker}, and we followed the same approach. We selected $d=40$ features for each speaker. This classical procedure, takes as input an audio file, and return a matrix where each row represent a point in $\mathbb{R}^{40}$, and each row represent a small window of audio file of a few milliseconds. Due to the differences in the speakers' audio data, the different dataset $V_1 \dots V_{38}$ are made of a variable number of points which ranges from $n = 2000$ to $4000$. Then, each speaker is modeled with a mixture of 16 Gaussians with diagonal covariance matrix. The test set consists of other 5 (or more) unseen utterances of the same 38 speakers. The task is to correctly label the unseen utterances with the name of the correct speaker. This is done by testing each of the GMM fitted during the training against the new test sample. The selected model is the one with the highest likelihood. In the experiments, we compared the performances of classical and quantum algorithm, and measured the relevant parameters that govern the runtime of the quantum algorithm. We used scikit-learn \citep{scikit-learn} to run all the experiments.

        We also simulated the impact of noise during the training of the the GMM fitted with ML estimate, so to assure the convergence of the quantum algorithm. For almost all GMM fitted using 16 diagonal covariance matrices, there is at least a $\Sigma_j$ with bad condition number (i.e up to 2500 circa). As in \citep{kerenidis2019q, KL18} we took a threshold on the matrix by discarding singular values smaller than a certain value. Practically, we discarded any singular value smaller than $0.07$. In the experiment, thresholding the covariance matrices not only did not made the accuracy worse, but had also a positive impact on the accuracy, perhaps because it has a regularizing effect on the model. For each of the GMM $\gamma^t$ estimated with ML estimate, we perturbed $\gamma$ at each iteration. Then, we measured the accuracy on the test set. For each model, the perturbation consists of three things. First we add to each of the components of $\theta$ some noise from the truncated gaussian distribution centered in $\theta_i$ in the interval $(\theta_i-\delta_\theta/\sqrt{k}, \theta_i+\delta_\theta/\sqrt{k} )$ with unit variance. This can guarantee that overall, the error in the vector of the mixing weights is smaller than $\delta_\theta$. 
        Then we perturb each of the components of the centroids $\mu_j$ with gaussian noise centered in $(\mu_j)_i$ on the interval $( (\mu_j)_i - \frac{\delta_\mu}{\sqrt{d}}), (\mu_j)_i+\frac{\delta_\mu}{\sqrt{d}})$. Similarly, we perturbed also the diagonal matrices $\Sigma_j$ with a vector of norm smaller than $\delta_\mu\sqrt{\eta}$, where $\eta=10$. As we are using a diagonal GMM, this reduces to perturbing each singular value with gaussian noise from a truncated gaussian centered $\Sigma$ on the interval $((\Sigma_j)_{ii}-\delta_\mu\sqrt{\eta}/\sqrt{d}, (\Sigma_j)_{ii}+\delta_\mu\sqrt{\eta}/\sqrt{d})$. Then, we made sure that each of the singular values stays positive, as covariance matrices are SPD. Last, the matrices are thresholded, i.e. the eigenvalues smaller than a certain threhosld $\sigma_\tau$ are set to $\sigma_\tau$. This is done in order to make sure that the effective condition number $\kappa(\Sigma)$ is no bigger than a threhosld $\kappa_\tau$. With a $\epsilon_\tau=7\times 10^{-3}$ and $70$ iterations per initialization, all runs of the $7$ different initialization of classical and quantum EM converged. Once the training has terminated, we measured all the values of $\kappa(\Sigma), \kappa(V), \mu(V), \mu(\Sigma), \log\det(\Sigma)$ for both ML and for MAP estimate. The results are in the Table on the main manuscript. Notably, thresholding the $\Sigma_j$ help to mitigate the errors of noise as it regularized the model. In fact, using classical EM with ML estimation, we reached an accuracy of $97.1\%$. With parameters of $\delta_\mu = 0.5$, $\delta_\theta=0.038$, and a threshold on the condition number of the covariance matrices of $\Sigma_j$ of $0.07$, we reached an accuracy of $98.7\%$. \footnote{The experiments has been improved upon a previous version of this work, by adding more data, adding the noise during the training procedure, and finding better hyperparameters. } 
        
        We further analyzed experimentally the evolution of the condition number $\kappa(V_i)$ while adding vectors from all the utterances of the speakers to the training set $V_i$. As we can see from Figure \ref{img:kappa}, all the condition numbers are pretty stable and do not increase by adding new vectors to the various training sets $V_1, \dots V_n$. 
        
        \begin{figure}[h]
        \caption{Evolution of $\kappa(V_i)$ where $V_i$ is the data matrix obtained by all the utterances available from the $i$-th speaker to the training set. For all the different speaker, the condition number of the matrix $V_i$ is stable, and does not increase while adding vectors to the training set.}\label{img:kappa}
        \centering
        \includegraphics[width=0.5\textwidth]{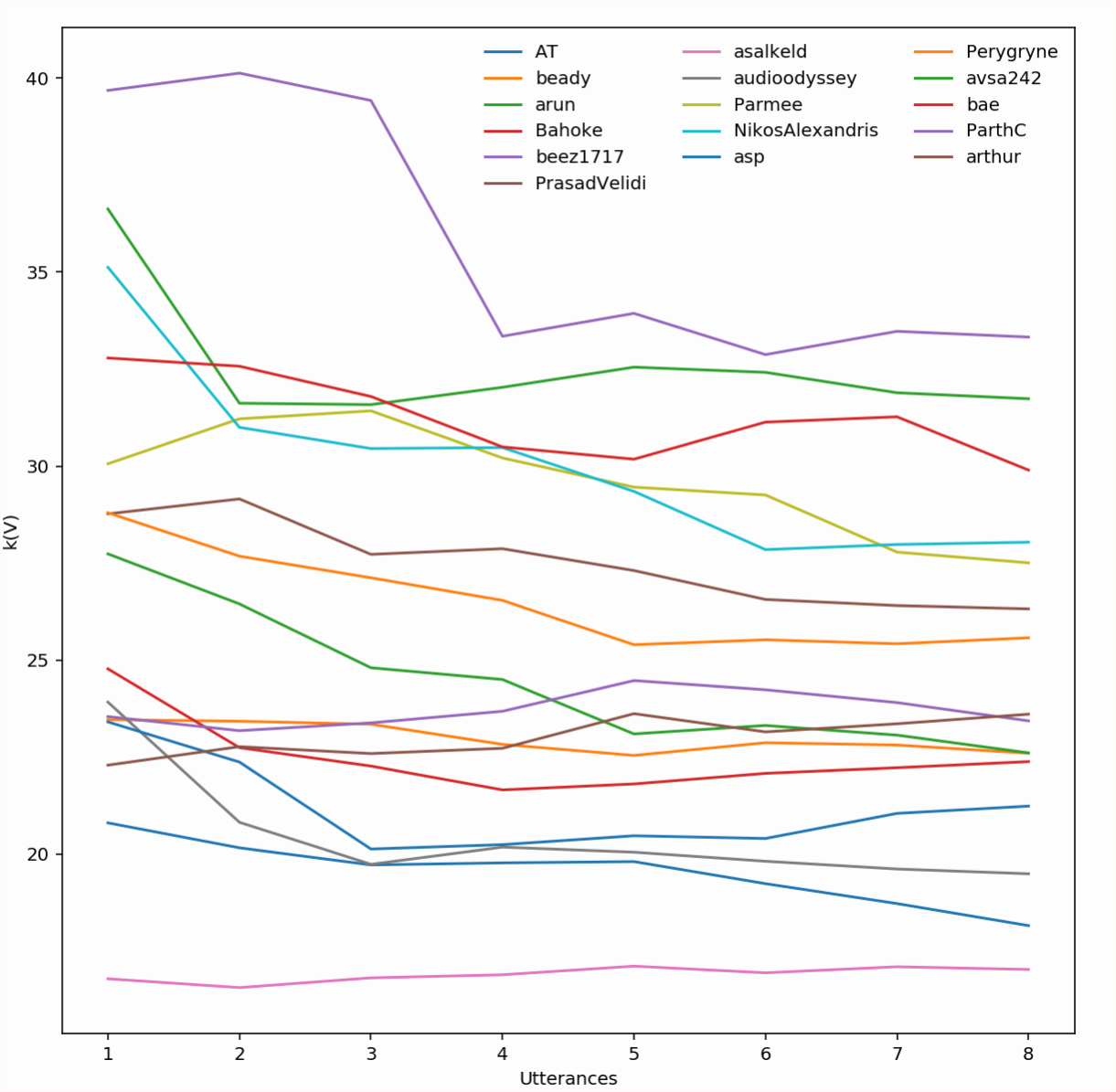}
        \end{figure}
        
        We leave for future work the task of testing the algorithm with further experiments (i.e. bigger and different types of datasets), and further optimizations, like procedures for hyperparameter tuning.

\section{Conclusions} 
Given the tremendous relevance of classical Expectation-Maximization algorithm, it is important to consider quantum versions of EM. Here we proposed a quantum Expectation-Maximization algorithm, and showed how to use it to fit a GMM in the ML estimation setting. We analyzed theoretically the runtime of QEM, and estimate it on real-world datasets, so to better understand cases where quantum computers can offer a computational advantage. While we discussed how to use QEM to fit other mixture models, its usage in the MAP settings is detailed in the Supp. Material. 

The experiments suggest that the influence of the extra parameters in the quantum running time is moderate. This makes us believe that, when quantum computers will be available, our algorithm could be useful in analyzing large datasets. We believe further improvements of the algorithm can reduce even more the complexity with respect to these extra parameters. For instance, the $\ell_\infty$ tomography can potentially remove the dependence on $d$ from the runtime. As we discussed above, we believe dequantization techniques can successfully be applied to this work as well, but we don't expect the time complexity of the final algorithm to be competitive with the classical EM algorithm or with this result. We leave for future work the study of quantum algorithms for fitting GMM robust to adversarial attacks \cite{xu2018parameter, pensia2019extracting,dasgupta1999learning}, and quantum algorithms for computing the log-determinant of Symmetric Positive Definite matrices, and further experiments.

\section{Acknowledgement}
This research has been partly supported by QuData, Quantex, and QuantAlgo. We thank Changpeng Shao for useful discussions. 

\bibliography{further_bibliography,Mendeley} 

\bibliographystyle{icml2020}

\appendix

\section{Supplementary Material}

We start by reviewing the classical EM algorithm for GMM, as in Algorithm \ref{GMMEM}.
\begin{algorithm}[ht]
	\caption{Expectation-Maximization for GMM} \label{GMMEM}
	\begin{algorithmic}[1]
		
		\REQUIRE  Dataset $V$, tolerance $\tau >0$.
		\ENSURE A GMM $\gamma^t= (\theta^t, \bm \mu^t, \bm \Sigma^t )$ that maximizes locally the likelihood $\ell(\gamma;V)$ up to tolerance $\tau$.
		\vspace{10pt} 
		\STATE Select $\gamma^0=(\theta^0, \bm \mu^{0}, \bm \Sigma^{0} )$ using classical initialization strategies described in Subsection \ref{initialization}. 
		\STATE $t=0$
		\REPEAT %
		\STATE {\bf Expectation}\\
		$\forall i,j$, calculate the responsibilities as:
		\begin{equation}
			r_{ij}^t = \frac{\theta^t_j \phi(v_i; \mu^t_j, \Sigma^t_j )}{\sum_{l=1}^k \theta^t_l \phi(v_i; \mu^t_l, \Sigma^t_l)}
		\end{equation}
		
		\STATE {\bf Maximization}\\
		Update the parameters of the model as:
		\begin{equation}
			\theta_j^{t+1} \leftarrow \frac{1}{n}\sum_{i=1}^n r^{t}_{ij}
		\end{equation} 
		
		\begin{equation}
			\mu_j^{t+1} \leftarrow \frac{\sum_{i=1}^n r^{t}_{ij} v_i }{ \sum_{i=1}^n r^{t}_{ij}}
		\end{equation} 
		
		\begin{equation}
			\Sigma_j^{t+1} \leftarrow \frac{\sum_{i=1}^n r^{t}_{ij} (v_i - \mu_j^{t+1})(v_i - \mu_j^{t+1})^T }{ \sum_{i=1}^n r^{t}_{ij}}
		\end{equation}

		\STATE t=t+1
		\UNTIL {$\: $}
		\STATE \begin{equation} | \ell(\gamma^{t-1};V) - \ell(\gamma^t;V) | < \tau \end{equation}
		
			\STATE Return $\gamma^{t}=(\theta^t, \bm \mu^t, \bm \Sigma^t )$
	\end{algorithmic}
\end{algorithm}

\subsection{Initialization strategies for EM}\label{initialization}
Unlike k-means clustering, choosing a good set of initial parameters for a mixture of Gaussian is by no means trivial, and in multivariate context is known that the solution is problem-dependent. There are plenty of proposed techniques, and here we describe a few of them. Fortunately, these initialization strategies can be directly translated into quantum subroutines without impacting the overall running time of the quantum algorithm. 

The simplest technique is called \emph{random EM}, and consists in selecting initial points at random from the dataset as centroids, and sample the dataset to estimate the covariance matrix of the data. Then these estimates are used as the starting configuration of the model, and we may repeat the random sampling until we get satisfactory results.

A more standard technique borrows directly the initialization strategy of \emph{k-means++}  proposed in \citep{arthur2007k}, and extends it to make an initial guess for the covariance matrices and the mixing weights. The initial guess for the centroids is selected by sampling from a suitable, easy to calculate distribution. This heuristic works as following: Let $c_0$ be a randomly selected point of the dataset, as first centroid. The other $k-1$ centroids are selected by selecting a vector $v_i$ with probability proportional to $d^2(v_i, \mu_{l(v_i)})$, where $
\mu_{l(v_i)}$ is the previously selected centroid that is the closest to $v_i$ in $\ell_2$ distance. 
These centroids are then used as initial centroids for a round of k-means algorithm to obtain $\mu_1^0 \cdots \mu_j^0$. Then, the covariance matrices can be initialized as $\Sigma_j^0 := \frac{1}{|\mathcal{C}_j|} \sum_{i \in \mathcal{C}_j } (v_i - \mu_j)(v_i - \mu_j)^T $, where $\mathcal{C}_j$ is the set of samples in the training set that have been assigned to the cluster $j$ in the previous round of k-means. The mixing weights are estimated as $\mathcal{C}_j/{n}$. Eventually $\Sigma_j^0$ is regularized to be a PSD matrix.

There are other possible choices for parameter initialization in EM, for instance, based on \emph{Hierarchical Agglomerative Clustering (HAC)} and the \emph{CEM} algorithm. In CEM we run one step of EM, but with a so-called classification step between E and M. The classification step consists in a hard-clustering after computing the initial conditional probabilities (in the E step). The M step then calculates the initial guess of the parameters \citep{celeux1992classification}. In the \emph{small EM} initialization method we run EM with a different choice of initial parameters using some of the previous strategies. The difference here is that we repeat the EM algorithm for a small number of iterations, and we keep iterating from the choice of parameters that returned the best partial results. For an overview and comparison of different initialization techniques, we refer to \citep{blomer2013simple, biernacki2003choosing}.

\paragraph{Quantum initialization strategies}\label{qheuristics}
For the initialization of $\gamma^0$ in the quantum algorithm we can use the same initialization strategies as in classical machine learning. For instance, we can use the classical \emph{random EM} initialization strategy for QEM.

A quantum initialization strategy can also be given using the $\emph{k-means++}$ initializion strategy from \citep{kerenidis2019q}, which returns $k$ initial guesses for the centroids $c_1^0 \cdots c_k^0$ consistent with the classical algorithm in time $\left(k^2 \frac{2\eta^{1.5}}{\epsilon\sqrt{\mathbb{E}(d^2(v_i, v_j))}}\right)$, where $\mathbb{E}(d^2(v_i, v_j))$ is the average squared distance between two points of the dataset, and $\epsilon$ is the tolerance in the distance estimation. From there, we can perform a full round of q-means algorithm and get an estimate for $\mu_1^0 \cdots \mu_k^0$. With q-means and quantum access to the centroids, we can create the state 
\begin{equation}\label{labels} 
\ket{\psi^0} := \frac{1}{\sqrt{n}}\sum_{i=1}^{n} \ket{i} \ket{ l(v_{i})}.
\end{equation}
Where $l(v_i)$ is the label of the closest centroid to the $i$-th point. By sampling $S \in O(d)$ points from this state we get two things. First, from the frequency $f_j$ of the second register we can have a guess of $\theta_j^0 \leftarrow |\mathcal{C}_j|/n \sim f_j/S$. Then, from the first register we can estimate $\Sigma_j^0 \leftarrow \frac{}{}\sum_{i\in S}(v_i - \mu_j^0)(v_i - \mu_j^0)^T$. Sampling $O(d)$ points and creating the state in Equation \eqref{labels} takes time $\widetilde{O}(dk\eta)$ by Theorem \ref{th:innerproductestimation} and the minimum finding procedure described in \citep{kerenidis2019q}. 

Techniques illustrated in \citep{miyahara2019expectation} can also be used to quantize the CEM algorithm which needs a hard-clustering step. Among the different possible approaches, the \emph{random} and the \emph{small EM} greatly benefit from a faster algorithm, as we can spend more time exploring the space of the parameters by starting from different initial seeds, and thus avoid local minima of the likelihood.

\subsection{Special cases of GMM.} What we presented in the main manuscript is the most general model of GMM. For simple datasets, it is common to assume some restrictions on the covariance matrices of the mixtures. The translation into a quantum version of the model should be straightforward. We distinguish between these cases:

\begin{enumerate}
    \item \textbf{Soft $k$-means}. This algorithm is often presented as a generalization of k-means, but it can actually be seen as special case of EM for GMM - albeit with a different assignment rule. In soft $k$-means, the assignment function is replaced by a softmax function with \emph{stiffness} parameter $\beta$. This $\beta$ represents the covariance of the clusters. It is assumed to be equal for all the clusters, and for all dimensions of the feature space. Gaussian Mixtures with constant covariance matrix (i.e. $\Sigma_j = \beta I$ for $\beta \in \mathbb{R}$) can be interpreted as a kind of soft or fuzzy version of k-means clustering. The probability of a point in the feature space being assigned to a certain cluster $j$ is:
    
$$ r_{ij}=\frac{e^{-\beta \norm{x_i - \mu_i}^2}}{\sum_{l=1}^k e^{-\beta \norm{x_i - \mu_l}^2  }} $$

where $\beta>0$ is the stiffness parameter. 

    \item \textbf{Spherical}. In this model, each component has its own covariance matrix, but the variance is uniform in all the directions, thus reducing the covariance matrix to a multiple of the identity matrix (i.e. $\Sigma_j = \sigma_j^2 I$ for $\sigma_j \in \mathbb{R}$).

    \item \textbf{Diagonal}. As the name suggests, in this special case the covariance matrix of the distributions is a diagonal matrix, but different Gaussians might have different diagonal covariance matrices.

    \item \textbf{Tied}. In this model, the Gaussians share the same covariance matrix, without having further restriction on the Gaussian. 
    \item \textbf{Full}. This is the most general case, where each of the components of the mixture have a different, SDP, covariance matrix.

\end{enumerate}

\subsection{Quantum MAP estimate of GMM}
Maximum Likelihood is not the only way to estimate the parameters of a model, and in certain cases might not even be the best one. For instance, in high-dimensional spaces, it is pretty common for ML estimates to overfit. 
Moreover, it is often the case that we have prior information on the distribution of the parameters, and we would like our models to take this information into account. These issues are often addressed using a Bayesian approach, i.e. by using a so-called Maximum A Posteriori estimate (MAP) of a model \citep[Section 14.4.2.8]{murphy2012machine}. MAP estimates work by assuming the existence of a \emph{prior} distribution over the parameters $\gamma$. The posterior distribution we use as objective function to maximize comes from the Bayes' rule applied on the likelihood, which gives the posterior as a product of the likelihood and the prior, normalized by the evidence. More simply, we use the Bayes' rule on the likelihood function, as $p(\gamma;V) = \frac{p(V;\gamma)p(\gamma)}{p(V)}$. This allows us to treat the model $\gamma$ as a random variable, and derive from the ML estimate a MAP estimate:

\begin{equation}\label{mapestimate}
\gamma^*_{MAP} = \argmax_\gamma \sum_{i=1}^{n} \log p(\gamma | v_i)
\end{equation}

Among the advantages of a MAP estimate over ML is that it avoids overfitting by having a kind of regularization effect on the model \citep[Section 6.5]{murphy2012machine}. Another feature consists in injecting into a maximum likelihood model some external information, perhaps from domain experts. 
This advantage comes at the cost of requiring ``good'' prior information on the problem, which might be non-trivial. In terms of labelling, a MAP estimates correspond to a \emph{hard clustering}, where the label of the point $v_i$ is decided according to the following rule:
\begin{align}\label{argmaxmap}
y_i = \argmax_j r_{ij} = \argmax_j \log p(v_i|y_i = j;\gamma) + \nonumber \\
 \log p(y_i = j ;\gamma)  \end{align}

Deriving the previous expression is straightforward using the Bayes' rule, and by noting that the softmax is rank-preserving, and we can discard the denominator of $r_{ij}$ - since it does not depend on $\gamma$ - and it is shared among all the other responsibilities of the points $v_i$. Thus, from Equation \ref{mapestimate} we can conveniently derive Equation \ref{argmaxmap} as a proxy for the label. Fitting a model with MAP estimate is commonly done via the EM algorithm as well. The Expectation step of EM remains unchanged, but the update rules of the maximization step are slightly different. In this work we only discuss the GMM case, for the other cases the interested reader is encouraged to see the relevant literature.  For GMM, the prior on the mixing weight is often modeled using the Dirichlet distribution, that is $\theta_j \sim \text{Dir}(\bm \alpha)$. For the rest of parameters, we assume that the conjugate prior is of the form $p(\mu_j, \Sigma_j) = NIW(\mu_j, \Sigma_j | \bm m_0, \iota_0, \nu_0, \bm S_0)$, where $\text{NIW}(\mu_j, \Sigma_j)$ is the Normal-inverse-Wishart distribution. The probability density function of the NIW is the product between a multivariate normal $\phi(\mu|m_0, \frac{1}{\iota} \Sigma)$ and a inverse Wishart distribution $\mathcal{W}^{-1}(\Sigma|\bm S_0, \nu_0)$. NIW has as support vectors $\mu$ with mean $\mu_0$ and covariance matrices $\frac{1}{\iota}\Sigma$ where $\Sigma$ is a random variable with inverse Wishart distribution over positive definite matrices. NIW is often the distribution of choice in these cases, as is the conjugate prior of a multivariate normal distribution with unknown mean and covariance matrix. 
A shorthand notation, let's define $r_j = n\theta_j = \sum_{i=1}^n r_{ij}$. As in \citep{murphy2012machine}, we also denote with $\overline{x_j}^{t+1}$ and $\overline{S_j}^{t+1}$ the Maximum Likelihood estimate of the parameters $(\mu_j^{t+1})_{ML}$ and $(\Sigma_j^{t+1})_{ML}$. For MAP, the update rules are the following: 

\begin{equation}
\theta_j^{t+1} \leftarrow \frac{r_j +\alpha_j-1}{n + \sum_j \alpha_j - k}
\end{equation}

\begin{equation}
\mu_j^{t+1} \leftarrow \frac{r_j \overline{x_j}^{t+1} + \iota_0 \bm m_0 }{r_j + \iota_0}
\end{equation}

\begin{equation}
\Sigma^{t+1}_j \leftarrow \frac{\bm S_0 + \overline{S_j}^{t+1}+ \frac{\iota_0r_j}{\iota_0+r_j} ( \overline{x_j}^{t+1} - \bm m_0)( \overline{x_j}^{t+1} - \bm m_0)^T }{\nu_0 + r_k + d + 2}
\end{equation}

Where the matrix $\bm S_0$ is defined as:

\begin{equation}
\bm S_0 := \frac{1}{k^{1/d}} Diag(s_1^2, \cdots, s_d^2),
\end{equation}

where each value $s_j$ is computed as $s_j := \frac{1}{n}\sum_{i=1}^n (x_{ij} - \sum_{i=1}^n x_{ij}))^2$ which is the pooled variance for each of the dimension $j$. For more information on the advantages, disadvantages, and common choice of parameters of a MAP estimate, we refer the interested reader to \citep{murphy2012machine}. Using the QEM algorithm to fit a MAP estimate is straightforward, since once the ML estimate of the parameter is recovered from the quantum procedures, the update rules can be computed classically.

\begin{corollary}[QEM for MAP estimates of GMM]\label{th:qmap}
	We assume we have quantum access to a GMM with parameters $\gamma^t$. For parameters $\delta_\theta, \delta_\mu, \epsilon_\tau > 0$, the running time of one iteration of the Quantum Maximum A Posteriori (QMAP) algorithm algorithm is 

$$O(T_\theta + T_\mu + T_\Sigma + T_\ell),$$

for
\begin{eqnarray*}
T_\theta  & = & \widetilde{O}\left(k^{3.5} \eta^{1.5} \frac{ \kappa(\Sigma)\mu(\Sigma) }{\delta_\theta^2}  \right)
\\
 T_\mu & = & \widetilde{O}\left(    \frac{k d\eta \kappa(V) (\mu(V) + k^{3.5}\eta^{1.5}\kappa(\Sigma)\mu(\Sigma))}{\delta_{\mu}^3}  \right)
 \\
 T_\Sigma & = & \widetilde{O} \Big( \frac{kd^2 \eta\kappa^2(V)(\mu(V')+\eta^2k^{3.5}\kappa(\Sigma)\mu(\Sigma))}{\delta_{\mu}^3} \Big)	
\\
 T_\ell & = & \widetilde{O}\left( k^{1.5}\eta^{1.5}  \frac{\kappa(\Sigma)\mu(\Sigma)}{\errlikelihood^2} \right)
 \end{eqnarray*}
For the range of parameters of interest, the running time is dominated by $T_\Sigma$.

\end{corollary}

\end{document}